\documentclass[a4paper,UKenglish]{article}
 
\usepackage{microtype}
\usepackage{algorithm, algorithmic}
\usepackage{bm}
\usepackage{wrapfig}
\usepackage{framed}
\usepackage{refcount}
\usepackage{fullpage}
\usepackage{amsmath}
\usepackage{amsthm}
\usepackage{authblk}
\usepackage{amssymb}
\usepackage{graphics}
\usepackage{graphicx}
\usepackage{hyperref}

\newcommand{\mypara}[1]{\smallskip\noindent\textbf{#1.}}

\def \cA {\mathcal A}

\def \Pr {\mathbb P}
\def \N {\mathbb N}

\newcommand{\Order}{\mathrm{O}}
\newcommand{\Exp}{\mathbb E}
\DeclareMathOperator*{\avg}{avg}

\DeclareMathOperator*{\argmin}{arg\,min}

\theoremstyle{definition}
\newtheorem{observation}{Observation}
\newtheorem{definition}{Definition}
\newtheorem{theorem}{Theorem}
\newtheorem{corollary}{Corollary}
\newtheorem{lemma}{Lemma}

\title{Preemptive Online Partitioning of Sequences \footnote{C.~Konrad is supported by 
the Centre for Discrete Mathematics and its Applications (DIMAP) at Warwick University and by EPSRC award EP/N011163/1.}}

\author{Christian Konrad}
\author{Tigran Tonoyan}
\affil[1]{Department of Computer Science and DIMAP, University of Warwick, UK\\
  \texttt{c.konrad@warwick.ac.uk}}
\affil[2]{ICE-TCS, School of Computer Science, Reykjavik University, Iceland\\
  \texttt{ttonoyan@gmail.com}}

\begin{document}

\maketitle

\begin{abstract}
Online algorithms process their inputs piece by piece, taking irrevocable decisions for each data item. This model
is too restrictive for most partitioning problems, 
since data that is yet to arrive may render it impossible to extend partial partitionings to the entire data set reasonably well.

In this work, we show that preemption might be a potential remedy. We consider the problem of partitioning online sequences,
where $p-1$ separators need to be inserted into a sequence of integers that arrives online so as to create $p$ contiguous partitions of similar weight. While without preemption no algorithm with
non-trivial competitive ratio is possible, if preemption is allowed, i.e., inserted partition separators may be removed 
but not reinserted again, then we show that constant competitive algorithms can be obtained. Our contributions include:

We first give a simple deterministic $2$-competitive preemptive algorithm for arbitrary $p$ and arbitrary sequences. Our main 
contribution is the design of a highly non-trivial partitioning scheme, which, under some natural conditions and $p$ being 
a power of two, allows us to improve the competitiveness to $1.68$. We also show that the competitiveness of deterministic (randomized) algorithms is at least $\frac{4}{3}$ (resp. $\frac{6}{5}$).

For $p=2$, the problem corresponds to the interesting special case of preemptively guessing the center of a weighted request sequence. 
While deterministic algorithms fail here, we provide a randomized $1.345$-competitive algorithm for all-ones sequences and 
prove that this is optimal. For weighted sequences, we give a  $1.628$-competitive algorithm and a lower bound of $1.5$.

\end{abstract}

\section{Introduction}
Online algorithms receive their inputs sequentially piece by piece. For each incoming piece of data (=request), the 
algorithm takes an immediate and irrevocable decision on how to process it. Taking good decisions can be challenging or even impossible, 
since decisions are based only on the requests seen and choices taken so far, and, in particular, they cannot be based on future requests.
For many problems, taking a few bad decisions is forgivable and good competitive algorithms can nevertheless be designed (e.g. maximum matching \cite{kvv90}, 
bin packing \cite{u71}, $k$-server \cite{mms90}), while for other problems, even a single bad decision may make it impossible 
to obtain non-trivial solutions (e.g. the maximum independent set problem \cite{himt02}). Generally, a necessary condition for a problem to admit
good online algorithms is that solutions can be {\em incrementally} built by extending partial solutions.

In this paper, we are interested in whether data partitioning problems can be solved online. In an online data partitioning problem, the input data $X$ arrives online 
and is to be partitioned into $p$ parts by computing a partitioning function $\phi: X \rightarrow \{1, 2, \dots, p \}$  such that an application-specific cost function is 
optimized. Unfortunately, most partitioning problems
are inherently non-incremental and thus poorly suited to the online model, which may explain why the literature on online data 
partitioning is exceptionally scarce (see related works section). This raises the following research questions:
\begin{enumerate}
 \item Is online data partitioning with provable guarantees really hopeless to achieve? 
 \item If so, how can we minimally augment the power of online algorithms to render data partitioning possible?
\end{enumerate}

In this work, we address these questions with regards to the problem of {\em partitioning integer sequences}.
It is one of the simplest data partitioning problems and thus a good candidate problem to answer the questions raised.
We prove that even for this rather simple problem, non-trivial quality guarantees are indeed impossible to achieve in the online model. 
However, if we augment the online model with {\em preemption}, i.e., the ability to remove a previously inserted
element from the solution, or, in the context of data partitioning, the ability to merge a subset of current partitions by removing 
previously inserted partition separators, then competitive ratios of at most $2$ can be obtained. 

\mypara{Partitioning Online Integer Sequences}
Let $X$ be an integer sequence of length $n$, and $p \ge 2$ an integer. In the problem of partitioning integer sequences (abbreviated by $\textsc{Part}$),
the goal is to partition $X$ into $p$ contiguous blocks (by determining the position of $p-1$ partition separators) such that the 
maximum weight of a block, denoted the bottleneck value of the partitioning, is minimized. 

In the (non-preemptive) online model, parameter $p$ is given to the algorithm beforehand, and the sequence $X$ arrives online, integer by integer. When processing an integer,
the algorithm has to decide whether or not to place a partition separator after the current integer. We are interested in the {\em competitive ratio} \cite{be98} of
an algorithm, i.e., the (expected) ratio of the bottleneck value of the computed partitioning and the optimal bottleneck value of the input, maximized over 
all inputs. 
Placing no separator at all results in a single partition that is trivially $p$-competitive. We prove, however, that this is essentially 
best possible in the non-preemptive online model, even if the integer sequence is an all-ones sequence. 


Given this strong impossibility result, we then augment the online model with preemption. When processing the current integer of the input sequence, a preemptive online algorithm 
for \textsc{Part} is allowed to remove a previously placed separator, which results in the merging of the two adjacent partitions incident to the separator. 
The removed separator can then be reinserted again (however, only after the current integer). In this paper, we show that the additional flexibility gained 
through preemption allows us to obtain algorithms with competitive ratio at most $2$. 
Even though our original motivation for studying \textsc{Part} in the preemptive online model is the fact that
non-trivial algorithms cannot be obtained in the non-preemptive case, preemptive online algorithms for \textsc{Part} are
extremely space efficient (only the weights of partitions and positions of separators need to be remembered) and thus 
work well in a data streaming context for massive data sets. 

\mypara{Partitioning Continuous Online Flows}
Algorithms for \textsc{Part} have to cope with the following two difficulties: First, the weights in the input sequence may vary hugely,
which implies that algorithms cannot establish partitions of predictable weights. 
Second, algorithms need to find a way to continuously merge adjacent partitions while keeping the bottleneck value small. 
While the first point can be tackled via rounding approaches, the second constitutes the core
difficulty of \textsc{Part}. We define a problem denoted \textsc{Flow}, which abstracts away the varying weights of the integers and
allows us focus on the second point. \textsc{Flow} differs from \textsc{Part} in that the preemptive online algorithm is allowed to determine 
the weight of every incoming element (we now even allow for positive rational weights). The difficulty in 
\textsc{Flow} stems from the fact that the algorithm is not aware of the total weight of the input. While at a first glance this problem appears to be substantially easier than \textsc{Part}, we show that any algorithm for \textsc{Flow} can 
be used for \textsc{Part} while incurring an error term that depends on the ratio of the largest weight
of an element and the total weight of the input sequence. \textsc{Flow} can be seen as a continuous version of \textsc{Part} and
can be interpreted as the problem of partitioning a continuous online flow (details follow in the preliminaries section). 


\mypara{Our Results}
We first show that every algorithm for \textsc{Part} in the non-preemptive online model has an approximation ratio of $\Omega(p)$, even if 
the input is guaranteed to be a sequence of ones (\textbf{Theorem~\ref{thm:lbnonpreemp}}). We then turn to the preemptive model and consider the special case $p=2$ first,
which corresponds to preemptively guessing the center of a weighted request sequence. It is easy to see that every deterministic algorithm for the 
$p=2$ case has a competitive ratio of $2$. We then give a randomized $1.345$-competitive algorithm for unweighted sequences (\textbf{Theorem~\ref{T:identicalupper}}) 
and prove that this is best possible (\textbf{Theorem~\ref{T:identicallower}}). We extend this algorithm to weighted sequences and give a $1.628$-competitive 
algorithm (\textbf{Theorem~\ref{T:weightedp2}}) and a lower bound of $1.5$ on the competitiveness (\textbf{Theorem~\ref{T:weightedlower}}).

For general $p$, we first give a simple deterministic $2$-competitive online algorithm for \textsc{Part} 
(\textbf{Theorem~\ref{thm:generalpub1}}) and prove a lower bound of $4/3$ ($6/5$) on the competitiveness of every deterministic (resp. randomized) algorithm (\textbf{Theorem~\ref{thm:lbpart}}). We then turn to \textsc{Flow} 
and give a highly non-trivial deterministic partitioning scheme with competitive ratio $1.68$ 
(\textbf{Theorem~\ref{thm:scheme}}) for the case that $p$ is a power of two, which constitutes the main contribution of 
this paper. This scheme translates to \textsc{Part} while incurring a small additive term in the competitive ratio 
that stems from the varying weights in \textsc{Part}. We discuss extensions of our scheme to arbitrary values 
of $p$ and demonstrate experimentally that competitive ratios better than $2$ can still be obtained. Last, we give a lower bound of 
$1.08$ on the competitive ratio of every deterministic algorithm for \textsc{Flow} (\textbf{Theorem~\ref{thm:lbflow}}). 
Unlike the lower bounds for $\textsc{Part}$, this lower bound does not rely on the discrete properties of integers.

\mypara{Techniques}
Consider first the $p=2$ case and all-ones sequences, which corresponds to preemptively guessing the center of the request sequence. 
One potential technique is {\em reservoir sampling} \cite{v85}, which allows the sampling of a uniform random element while processing the input
sequence. It naturally suits the preemptive online model and can be used to place a separator at a uniform random position in the request sequence, 
giving a randomized algorithm with expected competitive ratio $1.5$. We show that the {\em randomized geometric guessing} technique allows us to improve on this bound: For a 
random $\delta \in (0,1)$ and a carefully chosen value $X$, reset the single separator to the current position every time the total weight seen 
so far equals $\lceil X^{\delta \cdot i}\rceil$, for $i \in \{1, 2, \dots \}$, giving a $1.345$-competitive algorithm. Via Yao's
principle, we prove that this algorithm is optimal. We then analyze essentially the same algorithm on weighted sequences and show that it is 
$1.632$-competitive. 

For general $p$, consider the following algorithm for all-ones sequences of unknown length $n$ (assume also that $p$ is even): 
First, fill all partitions with weight $w = 1$. Whenever partitions are entirely filled, merge them pairwise creating
$p/2$ partitions of weight $2w$ and then update $w \gets 2 w$. Then $w$ always constitutes the bottleneck value of this partitioning. 
After the merging, fill the empty partitions with weight $w$ and repeat. Note that the optimal 
bottleneck value $opt$ is bounded as $opt \ge \frac{n}{p}$. Since at least half of all partitions computed by the algorithm have weight $w$, 
it holds that $w \frac{p}{2} \le n$, which together with $opt \ge \frac{n}{p}$ implies $w \le 2 opt$, giving a $2$-competitive algorithm. 
The $2$-competitive algorithm given in this paper is based on the intuition provided and also works for sequences with arbitrary weights. 

To go beyond the competitiveness of $2$, consider the key moment that leads to the $2$-competitiveness of the above algorithm:
Just after merging all $p$ partitions of weight $w$ into $p/2$ partitions of weight $2w$, the competitive ratio is $2$. 
To avoid this situation, note that merging only a single pair of partitions would not help. Instead, after merging
two partitions, we need to guarantee that the new bottleneck value is substantially smaller than twice the current bottleneck value. 
This implies that the weights of 
the merged partitions cannot both be close to the current bottleneck value, and it is hence beneficial to establish and merge partitions with 
different weights. In fact, our $1.08$ lower bound for \textsc{Flow} makes use of this observation: If at some moment the competitive 
ratio is too good, then most partitions have similar weight, which implies that when these partitions have to be merged in the future 
the competitive ratio will be large.

From an upper bound perspective, a clever merging scheme is thus required, which establishes partitions of different weights and, in particular,
remains analyzable. We give such a scheme for $\textsc{Flow}$, when $p$ is a power of two. As an illustration, consider
the case $p = 4$ as depicted in Table~\ref{tab:flowscheme}. Recall that in \textsc{Flow} we are allowed to determine the weight
of the incoming elements. We first initialize the partitions with values $x, x^2, x^3, x^4$, 
for $x = 2^\frac{1}{4}$, and evolve the partitions as in the table. Note that, at every moment, all partition weights are different, but
at the same time never differ by more than a factor of $2$. A key property is that at the end of the scheme the weights of 
the partitions are a multiple of the initial weights of the partitions. This allows us to repeate the scheme and limits the analysis
to one cycle of the scheme. We prove that the competitive ratio of our scheme is $1.68$ for every $p$ that is 
a power of two. We discuss how our scheme can be extended to other values of $p$ and demonstrate experimentally 
that competitive ratios better than $2$ can still be obtained.

\begin{table}
\begin{center}
\begin{tabular}{c|c|c|c|c}
$P_1$ & $P_2$ & $P_3$ &  $P_4$ & next \\
\hline 
 $\bm{x}$ & $\bm{x^{2}}$ & $x^{3}$ & $x^{4}$ & $x^{2}$ \\
 $x(1+x)$ & $x^{3}$ & $x^{4}$ & $\bm{x^{2}}$ & $\bm{x^{3}}$ \\
 $x(1+x)$ & $\bm{x^{3}}$ & $\bm{x^{4}}$ & $x^{2}(1 + x)$ & $x^{4}$ \\
 $x(1+x)$ & $x^{3}(1+x)$ & $x^{2}(1 + x)$ & $\bm{x^{4}}$ & $\bm{x^{5}}$ \\
 $\bm{x(1+x)}$ & $\bm{x^{3}(1+x)}$ & $x^{2}(1 + x)$ & $x^{4}(1+x)$ & $x^{3}(1+x)$ \\
 $x(1+x)(1+x^{2})$ & $\bm{x^{2}(1 + x)}$ & $\bm{x^{4}(1+x)}$ & $x^{3}(1+x)$ & $x^{5}(1+x)$ \\
 $x(1+x)(1+x^{2})$ & $x^{2}(1 + x)(1+x^{2})$ & $\bm{x^{3}(1+x)}$ & $\bm{x^{5}(1+x)}$ & $x^{4}(1+x)$\\
 $x(1+x)(1+x^{2})$ & $x^{2}(1 + x)(1+x^{2})$ & $x^{3}(1+x)(1+x^{2})$ & $\bm{x^{4}(1+x)}$ & $\bm{x^{6}(1+x)}$\\
 $x(1+x)(1+x^{2})$ & $x^{2}(1 + x)(1+x^{2})$ & $x^{3}(1+x)(1+x^{2})$ & $x^{4}(1+x)(1+x^{2})$ &
\end{tabular}
\end{center}
\caption{A partitioning scheme for $\textsc{Flow}$ for $p=4$. The bold elements are merged next.\label{tab:flowscheme}}
\end{table}

\mypara{Further Related Work}
The study of \textsc{Part} in the offline setting has a rich history with early works dating back to the 
80s \cite{b88,f91,hl92,mo95,ms95,kms97,hcn92,mp97,pa04}. After a series of improvements, Frederickson gave a highly 
non-trivial linear time algorithm \cite{f91}. \textsc{Part} finds many applications, especially in load balancing scenarios (e.g. \cite{mp97,k11,k16}). 
It has recently been studied in the context of streaming algorithms where it serves as a building block for partitioning XML documents \cite{k16}. 

Recently, Stanton and Kliot \cite{sk12} expressed interest in simple online strategies for data partitioning. 
They studied online graph partitioning heuristics\footnote{Phrased in the context of streaming algorithms, but their algorithms are in fact online} 
for the balanced graph partitioning problem and demonstrated experimentally that simple heuristics work well in practice. 
Stanton later analyzed the behavior of these heuristics on random graphs and gave good quality bounds \cite{s14}. 
Interestingly, besides this line of research, we are unaware of any further attempts at online data partitioning. 

Many works provide additional power to the online algorithm. Besides preemption, common resource augmentations include lookahead (e.g. \cite{g95}), distributions on the input (e.g. \cite{kmt11}), or advice (e.g. \cite{efkr11}). Preemptive online algorithms have been studied for various online problems. One example with a rich
history is the matching problem (e.g. \cite{fkmsz05,z12,elsw13,ctv15}). 

\mypara{Outline}
In Section~\ref{sec:prelim}, we formally define the studied problems and the preemptive online model. In Section~\ref{app:nonpreemptive}, we prove that 
every non-preemptive algorithm for \textsc{Part} has a competitive ratio of $\Omega(p)$. Then, 
we give our algorithms and lower bounds for $\textsc{Part}$ for the special case $p=2$ in Section~\ref{sec:pequals2}. 
All our results for \textsc{Part} and \textsc{Flow} for arbitrary $p$ are given in Section~\ref{sec:anyp}. 

\section{Preliminaries} \label{sec:prelim}
\mypara{Partitioning Integer Sequences}
In this paper, we study the following problem:
\begin{definition}[Partitioning Integer Sequences]
Let $X=w_1, \dots, w_n \in \mathbb{N}^n$ be an integer sequence, and let $p \in \mathbb{N}$ be an integer. The problem of partitioning integer sequences
consists of finding separators $S = s_1, \dots, s_{p-1}$ such that $1 = s_0 \le s_1 \le s_2 \le \dots \le s_{p-1} \le s_p = n+1$ and the maximum weight of a 
partition is minimized, i.e., 
$\max_{j \in \{0, \dots, p-1 \}} \sum_{i=s_j}^{s_{j+1}-1} w_i$
is minimized. The weight of a heaviest partition is the {\em bottleneck value} of the partitioning. This problem is abbreviated by $\textsc{Part}$.  
\end{definition}
\mypara{Online Model} In the online model, parameter $p$ is given to the algorithm beforehand, and the integers $X = w_1, \dots, w_n$ arrive online. 
Upon reception of an integer $w_i$ (also called a {\em request}), the algorithm has to decide whether to place a partition separator after $w_i$. 
In the non-preemptive online model, placing a separator is final, while in the preemptive model, when processing $w_i$ previously placed 
separators may be removed. The total number of separators in place never exceeds $p-1$ and separators can only be inserted at
the current request. 

\mypara{Competitive Ratio} The {\em competitive ratio} of a deterministic online algorithm for \textsc{Part} is the ratio between 
the bottleneck value of the computed partitioning and the bottleneck value of an optimal partitioning, maximized over all potential inputs. 
If the algorithm is randomized, then we are interested in the {\em expected competitive ratio}, where the expectation is taken over the 
coin flips of the algorithm.
 
 
\mypara{Partitioning Online Flows}
 We connect \textsc{Part} to the problem of partitioning a continuous online flow, abbreviated by \textsc{Flow}. 
 \begin{definition}[Preemptive Partitioning of a Continuous Online Flow]
In \textsc{Flow}, time 
is continuous starting at time $0$. Flow enters the system with unit and constant speed such that at time $t$, the total volume of 
flow $\int_{0}^t 1 \, dx = t$ has been injected. The goal is to ensure that when the flow stops at time $t_{max} \ge t_0$ 
($t_0$ is an arbitrarily small initial warm-up period), which is unknown to the algorithm, the total amount of flow is partitioned into 
$p$ parts such that the weight of a heaviest partition is minimized. More formally, at time $t_{max} \ge t_0$, the objective is that partition 
separators $S = s_1, \dots, s_{p-1}$ with $0 = s_0 \le s_1 \le \dots \le s_{p-1} \le s_p = t$ are in place such that:
$\max_{j \in \{0, \dots, p-1 \}} \int_{s_j}^{s_{j+1}} 1 \, dt = \max_{j \in \{0, \dots, p-1 \}} s_{j+1} - s_j$
is minimized. Similar to \textsc{Part}, in the preemptive online model, partition separators can only be inserted at the current time $t$, 
and previously inserted partition separators can be removed. 
\end{definition}
Even though \textsc{Flow} is defined as a continuous problem, it can be seen as a special case of \textsc{Part} where the algorithm can 
determine the weight of every incoming element.

\section{An $\Omega(p)$ Lower Bound In The Non-preemptive Online Model} \label{app:nonpreemptive}
Solving \textsc{Part} in the non-preemptive online model is difficult since  the total weight
of the input sequence is unknown to the algorithm. Since inserted partition separators cannot be removed, any partition 
created by the algorithm may be too small if the input sequence is heavier than expected. This intuition is formalized in the following theorem:

\begin{theorem}\label{thm:lbnonpreemp} 
 For every $p \ge 2$, every randomized non-preemptive online algorithm for \textsc{Part} has expected competitive ratio $\Omega(p)$,
 even on all-ones sequences.
\end{theorem}
\begin{proof}
Let $k = p^2$ and let $\Sigma = \{\sigma_1, \dots, \sigma_k \}$ be a set of request sequences where $\sigma_i$ is the all-ones sequence of length $4 ip$.
Let $\cA$ be a randomized algorithm for $\textsc{Part}$, and assume that its expected competitive ratio on every instance of $\Sigma$ is at most $c$. 
Then by Yao's lemma, there exists a deterministic algorithm $\cA_d$ with average approximation ratio at most $c$ over the instances of $\Sigma$. 

Let $S = \{s_1, \dots, s_{p-1} \}$ denote the set of separators output by $\cA_d$ on $\sigma_k$. Note that since $\sigma_i$ is a prefix 
of $\sigma_k$, the output of $\cA_d$ on $\sigma_i$ is a subset of the separators $S$. Partition now $\Sigma$ into $\Sigma_0$ and $\Sigma_1$ such that the 
approximation ratio of $\cA_d$ on $\sigma \in \Sigma_1$ is strictly smaller than $p/2$, and the approximation ratio of $\cA_d$ on $\sigma \in \Sigma_0$ is at 
least $p/2$. Consider now a $\sigma_i \in \Sigma_1$. Then, there exists a separator $s(i) \in S$ with $1/4 \cdot 4 i p \le s_j \le (3/4) \cdot 4 i p$, since otherwise the 
bottleneck value of the partitioning output by $\cA_d$ on $\sigma_i$ was at least $2 p i$. This in turn would imply that the approximation ratio was at least 
$2 p i / (4i) = \frac{1}{2} p$ (since the optimal bottleneck value on $\sigma_i$ is $4i$), contradicting the fact that $\sigma \in \Sigma_1$.
Next, note that the separators $s(i)$ and $s(j)$, for $i \neq j$ and $\sigma_i, \sigma_j \in \Sigma_1$, are necessarily different. Thus, since the number of separators 
is $p-1$, the size of $\Sigma_1$ is bounded by $p-1$. The average approximation factor $c$ of $\cA_d$ on instances $S$ is thus at least
$$c \ge \frac{|\Sigma_0| \cdot \frac{p}{2} + |\Sigma_1| \cdot 1}{|\Sigma|} = \frac{(k - (p-1)) \frac{p}{2}}{k} > \frac{p}{2} - \frac{p^2}{2k} \ge \frac{p}{2} - \frac{p^2}{2k} = \frac{p}{2} - \frac{1}{2} \ .  $$
\end{proof}

\section{Guessing the Center: Part for $p=2$} \label{sec:pequals2}
In this section, we consider \textsc{Part} for $p=2$, i.e., a single separator needs to be introduced into the request sequence
so as to split it into two parts of similar weight. We start with all-ones sequences and present an 
asymptotically optimal preemptive online algorithm. Then, we show how to extend this algorithm to sequences of arbitrary weights.

\subsection{All-ones Sequences}
The special case $p=2$ on all-ones sequences corresponds to preemptively guessing the center of a request sequence of unknown length. 
Deterministic algorithms cannot achieve a competitive ratio better than $2$ here, since request sequences that end just after a deterministic 
algorithm placed a separator give a competitive ratio of $2$.
\begin{observation}
 Every deterministic preemptive online algorithm for \textsc{Part} with $p=2$ has a competitive ratio of $2$.
\end{observation}
Using a single random bit, the competitiveness can be improved to $1.5$. This {\em barely-random} algorithm is 
given in Appendix~\ref{app:barelyrandom}. Using $O(\log n)$ random bits, we can improve the competitive ratio 
to $1.344$, which is best possible, and will be presented now.

\begin{algorithm}[H]
\caption{$\cA_x$}
 \begin{algorithmic}
   \STATE Choose uniform random $\delta\in (0,1)$, $i \rightarrow 0$
	\FOR{each request $j=1, \dots, n$}
	\IF {$j = \lceil x^{i+\delta} \rceil$}
	\STATE move separator to current position 
	\STATE $i\rightarrow i+1$
	\ENDIF
	\ENDFOR
 \end{algorithmic}
 \caption{Alg. $\cA_x$ for \textsc{Part} and $p=2$ \label{alg:p2}}
\end{algorithm}
Algorithm $\cA_x$, as depicted in Algorithm~\ref{alg:p2}, is parametrized by a real $x>2$, which will be optimized in Theorem~\ref{T:identicalupper}. 
It moves the separator to the current position as soon as the $\lceil x^{i+\delta} \rceil$-th request is processed, where $i$ is any integer and 
$\delta\in (0,1)$ is a random number.

\textbf{Remark.} The continuous random variable $\delta$ is only taken for convenience in the analysis; a bit precision of $O(\log n)$ 
is enough. 

In the following, denote by $R_{\cA_x}^n$ the competitive ratio of $\cA_x$ on a sequence of length $n$.

\begin{theorem}\label{T:identicalupper}
There is a constant $x\approx 3.052$ such that $\Exp [R_{\cA_x}^n] \approx 1.344 + O(n^{-1})$.
\end{theorem}
\begin{proof}
 Let $\alpha \in [0, 1)$ and $i \in \N$ be such that $n=2x^{i+\alpha}$.
Then, the bottleneck value of an optimal partition is $\lceil \frac{n}{2} \rceil = \lceil x^{i+\alpha} \rceil$.
We now bound the bottleneck value of the computed partitioning $r_{\cA_x}^n$, which depends on various ranges of $\alpha$ and $\delta$.
We distinguish two ranges for $\alpha$, and within each case, we distinguish three ranges of $\delta$:

\vspace{0.1cm}
\noindent \textbf{Case 1:} $\alpha>1-\log_x {2}$ (note that we assumed $x>2$). 
In order to bound $r_{\cA_x}^n$, we split the possible values of $\delta$ into three subsets:
 \begin{itemize}
\item{If $\delta\in (0,\alpha + \log_x 2 -1]$, then we have that $x^{\delta+i+1} \le 2x^{i+\alpha} = n$. In this case, the bottleneck value is $r_{\cA_x}^n=x^{\delta+i+1}+O(1)$. }
\item{If $\delta\in (\alpha + \log_x 2 -1, \alpha]$, then we have that $x^{\delta+i+1} > n$ but $x^{\delta+i} \le \frac{n}{2}$. In this case, $r_{\cA_x}^n=n - x^{\delta+i} + O(1)=2x^{\alpha+i}-x^{\delta+i} + O(1)$.}
\item{If $\delta\in (\alpha, 1)$, then we have that $x^{\delta+i+1} > n$ and $x^{\delta+i} \in (\frac{n}{2},n)$. In this case, $r_{\cA_x}^n=x^{\delta+i}$.}
\end{itemize}
Using these observations, we can bound the expected competitive ratio as follows:
\begin{align*}
\Exp[R_{\cA_x}^n]&=\int_0^{\alpha + \log_x 2 -1} {\frac{x^{\delta+i+1}}{x^{i+\alpha}}} d\delta +
 \int_{\alpha + \log_x 2 -1}^\alpha {\frac{2x^{\alpha+i} - x^{\delta+i}}{x^{i+\alpha}}} d\delta +
 \int_\alpha^1 {\frac{x^{\delta+i}}{x^{i+\alpha}}} d\delta + O(n^{-1}) \\
&= \frac{1}{x^\alpha}\left(\int_0^{\alpha + \log_x 2 -1} {x^{\delta+1}} d\delta +
\int_{\alpha + \log_x 2 -1}^\alpha (2x^{\alpha} - x^{\delta}) d\delta +
 \int_\alpha^1 {x^{\delta}} d\delta \right) + O(n^{-1}) \\
&=2-2\log_x 2 +\frac{2}{x\ln x} + O(n^{-1}). 
\end{align*}

\vspace{0.1cm}

\noindent \textbf{Case 2:} $\alpha\le 1-\log_x {2}$.
We deal with this case similarly, but we need to group the possible values for $\delta$ in a different way:
\begin{itemize}
\item{If $\delta\in (0,\alpha]$, then $x^{\delta+i+1} > n$ but $x^{\delta+i} \le \frac{n}{2}$. In this case, $r_{\cA_x}^n=n - x^{\delta+i} + O(1)$.}
\item{If $\delta\in (\alpha, \alpha + \log_x 2]$, then $x^{\delta+i} > \frac{n}{2}$ and $x^{\delta+i} \le n$. In this case, $r_{\cA_x}^n=x^{\delta+i} + O(1)$.}
\item{If $\delta\in (\alpha + \log_x 2, 1)$, then $x^{\delta+i} > n$. In this case, $r_{\cA_x}^n=n - x^{\delta+i-1} + O(1)$ (note that $i\ge 1$ here).}
\end{itemize}

Plugging the values above in the formula for the expected value, we obtain a different sum of integrals, which however leads to the same function as above:
\begin{align*}
\Exp[R_{\cA_x}]&=
\int_0^{\alpha} {\frac{2x^{\alpha+i} - x^{\delta+i}}{x^{i+\alpha}}} d\delta +
 \int_{\alpha}^{\alpha + \log_x 2} {\frac{x^{\delta+i}}{x^{i+\alpha}}} d\delta +
 \int_{\alpha + \log_x 2}^1 {\frac{2x^{\alpha+i} - x^{\delta+i-1}}{x^{i+\alpha}}} d\delta + O(n^{-1})\\
&=2-2\log_x 2 +\frac{2}{x\ln x} + O(n^{-1}). 
\end{align*}
Moreover, the formulas above are independent of $\alpha$. Thus, it remains to find a value of $x$ that minimizes $f(x)\overset{def}{=} 2-2\log_x 2 +\frac{2}{x\ln x}$. Observe that $f'(x)=-\frac{2}{x^2\ln^2 x} - \frac{2}{x^2\ln x}+\frac{\ln 2}{x\ln^2 x}$, and $f'(x)=0$ if and only if $x=\log_2(ex)$. With a simple transformation, the latter is equivalent to $z e^z=-\frac{\ln 2}{e}$ with $z=-x\ln 2$, so the value that minimizes $f(x)$ can be computed as 
$x_{\text{min}}=-\frac{W_{-1}(-\ln 2/e)}{\ln 2}\approx 3.052$, where $W_{-1}$ is the lower branch of Lambert's $W$ function. The claim of the theorem follows 
by calculating $f(x_{\text{min}})$.
\end{proof}

Next, we prove that no algorithm can achieve an (expected) competitive ratio better than the one claimed in Theorem~\ref{T:identicalupper}.
The proof applies Yao's Minimax principle and uses a hard input distribution over all-ones sequences of length $n \in [n_{\text{min}},n_{\text{max}}]$, 
for some large values of $n_{\text{min}}$ and $n_{\text{max}}$, where the probability that the sequence is of length $n$ is proportional to $1/n$. 

\begin{theorem}\label{T:identicallower}
For any randomized algorithm $\cA$, $\Exp[R_{\cA}^n]\ge 1.344$.
\end{theorem}
\begin{proof}

\def \nmin {{n_\text{min}}}
\def \nmax {{n_\text{max}}}

We will prove the theorem by using Yao's Minimax principle \cite{y77}. To this end, let us first 
consider an arbitrary \emph{deterministic} algorithm $\cA_{\text{det}}$. Assume the length of the 
sequence is random in the interval $X := [\nmin,\nmax]$ for large values of 
$\nmax$ and $\nmin$ with $\nmax > 2 \cdot \nmin$ and has the following distribution: 
The sequence ends at position $n\in X$ with probability $p_n$ which is proportional to $\frac{1}{n}$, i.e., using the
definition $\displaystyle S=\sum_{m=\nmin}^{\nmax}\frac{1}{m}$, we have
$$p_n := \Pr[\text{sequence is of length }n]= \frac{1}{n\cdot S} \ .$$

We will show that for \emph{each} deterministic algorithm $\cA_{\text{det}}$, if the 
input sequence is distributed as above, then $\Exp [R_{\cA_{\text{det}}}^n] \ge 1.344 - O(\ln^{-1} \nmax/\nmin)$, where the expectation 
is taken over the  distribution of $n$. 

Let $J$ denote the set of requests at which $\cA_{\text{det}}$ places the separator when processing the 
all-ones sequence of length $\nmax$. Note that the set of separator positions placed by $\cA_{\text{det}}$ 
on sequences of shorter lengths are a subset of $J$. Let $I = J \cap X = \{i_1, \dots, i_k \}$ (the $i_j$ are ordered with increasing value). 

For $n \in X$, let $r_{\cA_{det}}^n$ be the bottleneck value of the partitioning computed by $\cA_{\text{det}}$ 
on the sequence of length $n$. We bound $\Exp [R_{\cA_{\text{det}}}^n] = \sum_{n = \nmin}^{\nmax} p_n   R_{\cA_{\text{det}}}^n$ by separately bounding every partial sum in the following decomposition:
$$
\Exp [R_{\cA_{\text{det}}}^n] =
E(\nmin, i_1) + E(i_1,i_2) + \dots + 
E(i_{k-1}, i_k) + E(i_k, \nmax), 
$$
where for each $a > b$,  $E(a, b)=\sum_{n = a}^{b - 1} p_n  R_{\cA_{\text{det}}}^n$.
The first and last terms need a special care, so we will start with bounding all other terms. In the following, $H_p^{q}=\sum_{n=p}^{q}\frac{1}{n}$ denotes partial harmonic sums for $q\ge p\ge 1$. In particular, $S=H_{\nmin}^\nmax$.

Thus, we proceed in three steps:
\begin{enumerate}
 \item Consider an index $1 \le j < k$ and let us bound the sum $E(i_j, i_{j+1})$. Let us denote $a=i_j$ and $b=i_{j+1}$.
We need to consider two cases.

\vspace{0.1cm}
\noindent \textbf{Case 1:} $b \le 2 a$. Then for all $n \in \{a,\dots,b-1 \}$, the bottleneck value computed by the algorithm is
$r_{\cA_{det}}^n=a$ (since $n/2 < a$). Then:
\begin{equation}\label{E:p2lbsum}
E(a, b) \ge \sum_{n=a}^{b-1} \frac{1}{n S} \cdot \frac{a}{\lceil n/2\rceil}\ge \frac{2a}{S}\sum_{n=a}^{b-1} \frac{1}{n(n+1)}=\frac{2a}{S}\left(\frac{1}{a}-\frac{1}{b}\right)>1.4\cdot \frac{H_a^{b-1}}{S},
\end{equation}
where the last inequality can be proved as follows. First, it is easily checked that the ratio $\Phi(a,b)=2(1-\frac{a}{b})/H_{a}^{b-1}$ decreases 
when $b$ increases with $a$ kept fixed, implying that (recall that $b\le 2a$ and using standard approximations of harmonic sums) 

$$\Phi(a,b) \ge  \Phi(a,2a)= 1/H_{a}^{2a-1} \ge \frac{1}{\ln 2 + \frac{1}{a}}>1.4,$$
where the last inequality holds when $a=i_j$ is large enough (say $i_j\ge \nmin/2\ge 50$).

\vspace{0.1cm}

\noindent \textbf{Case 2:} $b > 2a$. In this case, for all $n=a,\dots,2a-1$, if the sequence is of length $n$, then $r_{\cA_{det}}^n=a$, as in case 1. 
However, when $n\ge 2a$, then $n/2 \ge a$, so $r_{\cA_{\text{det}}}^n=n-a$. Using these observations, we can bound $\Exp [R_{\cA_{\text{det}}}^n]$ as follows:
\begin{eqnarray*}
\Exp [R_{\cA_{\text{det}}}^n] & = & \sum_{n=a}^{2a-1} \frac{1}{nS} \frac{a}{\lceil n/2\rceil} + \sum_{n=2a}^{b-1} \frac{1}{n S} \frac{n-a}{\lceil n/2\rceil}  \\
& \ge & \frac{2a}{S}\sum_{n=a}^{2a-1}\frac{1}{n(n+1)} + \frac{2}{S}H_{2a + 1}^{b}  -\frac{2a}{S}\sum_{n=2a}^{b-1} \frac{1}{n(n+1)}\\
& = & \frac{2a}{S}\left(\frac{1}{a}-\frac{1}{2a}\right) + \frac{2}{S}H_{2a + 1}^{b} - \frac{2a}{S}\left(\frac{1}{2a}-\frac{1}{b}\right) \\
& = & \frac{2}{S} + \frac{2 a}{S b} + \frac{2}{S}(H_{a}^{b} - H_{a}^{2a}) \\
& = & \frac{H_{a}^{b}}{S} \cdot \left( 2 + \frac{2a}{bH_{a}^{b}} - \frac{H_a^{2a}}{H_{a}^{b}}\right),
\end{eqnarray*}
where the third line is obtained by using the identity $H_{2a+1}^b = H_a^b-H_a^{2a}$.
Again, using a standard approximation for the harmonic sums, and setting $x = \frac{b}{a}$, we can approximate:
\begin{eqnarray*}
2 + \frac{2a}{bH_{a}^{b}} - \frac{H_a^{2a}}{H_{a}^{b}} \ge 2 + \frac{2a}{b\ln \frac{b}{a}} - \frac{\ln 2}{\ln \frac{b}{a}} - \Order(a^{-1}) =2 + \frac{2}{x\ln x} -\log_x 2 - \Order(a^{-1}),
\end{eqnarray*}
Note that the function $f(x)=2 + \frac{2}{x\ln x} -\log_x 2$ is exactly the same 
that was minimized in the proof of Thm.~\ref{T:identicalupper}, and achieves its minimum in 
$(1,\infty)$ at $x_{\text{min}} \approx 3.052$, giving $f(x_{\text{min}})\approx 1.344$. Thus, we have 
$E(i_j, i_{j+1}) \ge \frac{H_{i_j}^{i_{j+1}-1}}{S}(1.344 -O(\frac{1}{i_j}))$.

\item The term $E(i_k, \nmax)$ can be bounded by $\frac{H_{i_k}^{\nmax-1}}{S}(1.344 -O(\frac{1}{i_k}))$ by an identical argument as above.

\item The term $E(\nmin, i_1)$ needs a slightly different approach. Let $i_0$ denote the last separator that the algorithm placed before $\nmin$. We can assume that $i_0\ge \nmin/2$, as otherwise the algorithm could only profit by moving the separator to $\nmin/2$. We consider two cases. First, if $i_1\le 2i_0$, then we simply assume the algorithm performs optimally in the range $[\nmin, i_1)$: 
\[
E(\nmin, i_1)=\frac{H_{\nmin}^{i_1-1}}{S} \ge 1.344 \cdot \frac{H_{\nmin}^{i_1-1}}{S} - 0.5\frac{H_{\nmin}^{i_1-1}}{S} > 1.344 \cdot \frac{H_{\nmin}^{i_1-1}}{S} - 1/S,
\]
since (recalling that $i_1\le 2i_0\le 2\nmin$),  $H_{\nmin}^{i_1} < H_{\nmin}^{2\nmin}< 1$.

On the other hand, when $i_1>2i_0$ (and by the discussion above, $2i_0 \ge \nmin$), then with calculations similar to the one in Case 2 above, we can obtain: 
\begin{align*}
E(x_\text{min}, i_1)&=\sum_{n=\nmin}^{2i_0-1} \frac{1}{nS} \frac{i_0}{\lceil n/2\rceil} + \sum_{n=2i_0}^{i_1-1} \frac{1}{n S} \frac{n-i_0}{\lceil n/2\rceil}\\
& \ge \frac{2i_0}{S}\sum_{n=\nmin}^{2i_0-1}\frac{1}{n(n+1)} + \frac{2}{S}H_{2i_0 + 1}^{i_1}  -\frac{2i_0}{S}\sum_{n=2i_0}^{i_1-1} \frac{1}{n(n+1)}\\
& \ge \frac{2}{S}\left(\frac{i_0}{\nmin} + \frac{i_0}{i_1} + H_{\nmin}^{i_1} - H_{\nmin}^{2i_0}\right) \ge 2\cdot \frac{H_{\nmin}^{i_1}}{S} - O(1/S),
\end{align*}
because $i_0\le \nmin$ and thus $H_{\nmin}^{2i_0} < 1$.
\end{enumerate}
It remains to plug the obtained estimates in (\ref{E:p2lbsum}):
\begin{align*}
\Exp [R_{\cA_{\text{det}}}^n] &=E(\nmin, i_1) + \sum_{j=1}^{k-1} E(i_j,i_{j+1}) + E(i_k, \nmax)\\
 &\ge (1.344 - O(1/\nmin))\cdot\frac{H_\nmin^{i_1} + \sum_{j=1}^{k-1} H_{i_j}^{i_{j+1}-1} + H_{i_k}^\nmax}{S}-O(1/S) \\
&= 1.344 - O(1/S).
\end{align*}
Last, by Yao's principle, every randomized algorithm has a competitive ratio of at least $\Exp [R_{\cA_{\text{det}}}^n]$.
\end{proof}

\subsection{General Weights}
Algorithm $\cA_x$ can be adapted to weighted sequences $X=w_1,w_2,\dots,w_n$ of positive integers as follows:  
$X$ can be thought of as a sequence of $W_n=\sum_{i=1}^n w_i$ unit weight requests and we simulate $\cA_x$ on this unit
weight sequence. Whenever $\cA_x$ attempts to place a separator, but the position does not fall at the end of a
weight $w_i$, the separator is placed after $w_i$.

If the weights of the sequence are bounded, algorithm $\cA_x$ can be analyzed similarly as Theorem~\ref{T:identicalupper}, 
by treating all requests as unit weights. This introduces an additional error term: 
\begin{corollary} \label{cor:lowarbitraryweights}
There is a constant value of $x\approx 3.052$ such that $E[R_{\cA_x}^n] \approx 1.344 + O(B/W_n)$ for any sequence $X=w_1,w_2,\dots,w_n$ with weights $w_i \le B$.
\end{corollary}



When arbitrary weights are possible, a non-trivial bound can still be proved. Interestingly, the optimal gap size between the separator positions
is larger than in the case of unweighted sequences. 

\begin{theorem}\label{T:weightedp2}
There is a constant value of $x\approx 5.357$ such that $\Exp [R_{\cA_x}^X] \le  1.627 + O(W_n^{-1})$ for each sequence of total weight $W_n$.
\end{theorem}
\begin{proof}
Let $X=w_1,w_2,\dots,w_n$ be the input sequence of total weight $W_n$, and let $m = \argmin_{m'} \sum_{i \le m'} w_i \ge \frac{W_n}{2}$. Then,
$w_m$ is the {\em central weight} of the sequence, and we denote $\frac{W_n}{2}$ the {\em central point}. We will argue first that replacing all 
$w_i$ left of $w_m$ by a sequence of $\sum_{i < m} w_i$ unit
requests, and replacing all $w_i$ right of $w_m$ by a single large request of weight $\sum_{i > m} w_i$ worsens the approximation factor of the 
algorithm. Indeed, suppose that the algorithm attempts to place a separator at a position $j$ that falls on an element $w_i$, which is located left
of $w_m$. Then the algorithm places the separator after $w_i$, which brings the separator closer to the center and thus improves the partitioning.
Similarly, suppose that the algorithm attempts to place a separator at position $j$ that falls on an element $w_i$, which is located right of $w_m$.
By replacing all weights located to the right of $w_m$ by a single heavy element, the algorithm has to place the separator at the end of the sequence, 
which gives the worst possible approximation ratio. Thus, we suppose from now on that $X$ is of the form $X = 1 \dots 1 w_m B$, where $B$ may be non-existent.

Assume that the central point splits the request with weight $w_m$ into two parts $w',w''\ge 0$, such that $w' + \sum_{i < m} w_i = \frac{W_n}{2}$. 
Clearly, the optimal bottleneck value is $opt=\frac{W_n}{2} + \min\{w',w''\}$. 
Further assume that $W_n=2x^{i+\alpha}$ for $\alpha\in [0,1)$ and $i\in \N$. Then, $\frac{W_n}{2}=x^{i+\alpha}$ is the central point. 
Let $\alpha_1, \alpha_2$ be such that $x^{i+\alpha_1}$ is the starting point of weight $w_m$, and 
$x^{i+\alpha_2}$ is the starting point of weight $B$. Note that $\alpha_2$ is non-negative, but $\alpha_1$ can be negative. However, if $\alpha_1<0$, we can replace $\alpha_1=0$ without decreasing the approximation ratio, because in both cases the algorithm places the separator after $w_m$, while in the case of  negative $\alpha_1$ the optimum can only be larger than when $\alpha_1=0$. Thus, we assume w.l.o.g. that $\alpha_1\ge 0$, so $\alpha_1\in [0, \alpha]$ and $\alpha_2\in [\alpha, \alpha+\log_x 2]$.


Again, we consider several cases. In the estimates below, we ignore the rounding terms as they are all 
$O(1)$, and hence the error term in the approximation ratio is $O(opt^{-1})=O(W_n^{-1})$. 
Also, we use $W$ in place of $W_n$ for brevity.

\vspace{0.1cm}
\noindent \textbf{Case 1:}  $\alpha< 1-\log_x {2}$. In this case, $x^{i+1}>W$. So for all $\delta\in [0,\alpha+\log_x 2]$, $x^{i+\delta}$ is the position where the last separator would be (in the unit weights case). However, when $\delta>\alpha+\log_x 2$ then $x^{i+\delta} > W$ and $x^{i+\delta-1}<W/2$ is the last separator.

Here we have:
\begin{itemize}
\item{$\delta\in [0,\alpha_1]$: $r_{\cA_x}= W-x^{i+\delta}$ (all unit weights before $w_m$).}
\item{$\delta\in [\alpha_1, \alpha_2]$: the separator is placed after $w_m$, and $r_{\cA_x} = W/2 + w''$.}
\item{$\delta\in [\alpha_2, \alpha+\log_x 2]$: assume the worst case bound $r_{\cA_x} \le W$.}
\item{$\delta\in [\alpha + \log_x 2, 1]$:  $r_{\cA_x} = W-x^{i+\delta-1}$, as the last separator is at $x^{i+\delta-1} < x^i \le x^{i+\alpha_1}$.}
\end{itemize}
Computing the expectation gives:
\begin{align}
E[R^X_{\cA_x}]&\le \int_0^{\alpha_1} \frac{W-x^{i+\delta}}{opt} d\delta +
 \nonumber \int_{\alpha_1}^{\alpha_2} {\frac{W/2 + w''}{opt}} d\delta 
 +\int_{\alpha_2}^{\alpha + \log_x 2} {\frac{W}{opt}} d\delta \\
\nonumber &+ \int_{\alpha + \log_x 2}^1 \frac{W-x^{i+\delta-1}}{opt} d\delta \\
&= \frac{1}{opt}\left(W(1+\frac{1}{x\ln x}) -\alpha_2(W/2-w'') + \alpha_1(W/2-w'')-\frac{x^{i+\alpha_1}}{\ln x} \right).\label{E:case11pre}
\end{align}
First, assume that $w'' \le w'$, which implies that $\alpha_1 \le \log_x (W/2 - w'')-i$. In this case $opt=W/2+w''$. Let us see which values of $\alpha_1$ maximize the term $\phi(\alpha_1)=(\frac{W}{2}-w'')\alpha_1 -\frac{x^{i+\alpha_1}}{\ln x}$ in the parentheses. We have $\phi'(\alpha_1)=\frac{W}{2}-w'' - x^{i+\alpha_1}\ge 0$, since $\alpha_1 \le \log_x (W/2 - w'')-i$ by assumption (that $w''\le w'$). So $\phi(\alpha_1)$ is increasing in the interval $[0,\log_x (W/2 - w'')-i]$, and $\alpha_1=\log_x (W/2 - w'')-i$ is the maximizer. Plugging this value and $\alpha_2=\log_{x}(W/2+w'')-i$ in the expectation formula and rearranging the terms gives:
\begin{align*}
E[R^X_{\cA_x}] \le \frac{W(1+\frac{1}{x\ln x}) + (\frac{W}{2}-w'')(\log_x (\frac{W}{2}-w'') - \log_x (\frac{W}{2}+w'') -\frac{1}{\ln x})}{\frac{W}{2}+w''}.
\end{align*}
Note that the right hand side is a decreasing function of $w''$ and the maximum is achieved when $w''=0$:
\begin{equation}\label{E:case111}
E[R^X_{\cA_x}] \le 2+\frac{2}{x\ln x}-\frac{1}{\ln x}.
\end{equation}
Now assume that $w''\ge w'$. Then $opt=W/2+w'$. By plugging  the value $\alpha_2=\log_x (W/2+w'')-i$ in (\ref{E:case11pre}) we see again that the expression is a decreasing function of $w''$, so it is maximum when $w''$ takes its minimum value $w''=w'$:
\begin{align*}
E[R^X_{\cA_x}]&\le 
\frac{W(1+\frac{1}{x\ln x}) +(\frac{W}{2}-w')(\log_x(\frac{W}{2}-w')-\log_x(\frac{W}{2}+w')-\frac{1}{\ln x})}{\frac{W}{2}+w'}.
\end{align*}
where we also used $x^{\alpha_1+i}=W/2-w'$. This is again a decreasing function of $w'$ and gives exactly the same bound (\ref{E:case111}).

\vspace{0.1cm}
\noindent \textbf{Case 2:} $\alpha > 1- \log_x 2$. Here $x^{i+\delta+1} < W$ for all $\delta\in [0,\alpha+\log_x 2-1]$, so we assume that the algorithm places the separator at the end of the sequence for such $\delta$.
Here we have to distinguish two sub-cases.

\vspace{0.1cm}
\noindent \textbf{Case 2.1:} $\alpha_1 > \alpha+\log_x 2-1$. We have:
\begin{itemize}
\item{$\delta\in [0,\alpha+\log_x 2-1]$: $r_{\cA_x}\le W$.}
\item{$\delta\in [\alpha+\log_x 2-1, \alpha_1]$:  $r_{\cA_x} = W-x^{i+\delta}$.}
\item{$\delta\in [\alpha_1, \alpha_2]$:  $r_{\cA_x} = W/2 + w''$.}
\item{$\delta\in [\alpha_2, 1]$:  $r_{\cA_x} \le W$.}
\end{itemize}

Computing the expectation gives:
\begin{align*}
E[R^X_{\cA_x}]&\le \int_0^{\alpha+\log_x 2-1} \frac{W}{opt} d\delta +
 \int_{\alpha+\log_x 2-1}^{\alpha_1} {\frac{W-x^{i+\delta}}{opt}} d\delta +
 \int_{\alpha_1}^{\alpha_2} {\frac{W/2 + w''}{opt}} d\delta + \int_{\alpha_2}^1 \frac{W}{opt} d\delta \\
&= \frac{1}{opt}\left(W(1+\frac{1}{x\ln x}) -\alpha_2(W/2-w'') + \alpha_1(W/2-w'')-\frac{x^{i+\alpha_1}}{\ln x} \right).
\end{align*}
The latter is the same formula as in Case 1, so (\ref{E:case111}) holds in this case too.

\vspace{0.1cm}
\noindent \textbf{Case 2.2:} $\alpha_1 \le \alpha+\log_x 2-1$. In the analysis below, we assume that $r_{\cA_x}$ takes the worst-case value $W$ for all $\delta\in [0,\alpha+\log_x 2-1]$. In this case, we assume w.l.o.g. that $\alpha_1=\alpha+\log_x 2-1$, because the optimum can only become worse when $\alpha_1$ is smaller, while the algorithm (as we assume) will not profit. Thus, $x^{i+\alpha_1}=W/x$ and $\alpha_1=\log_x W - 1 - i$.
We have:
\begin{itemize}
\item{$\delta\in [0,\alpha_1]$: $r_{\cA_x}\le W$.}
\item{$\delta\in [\alpha_1, \alpha_2]$:  $r_{\cA_x} = W/2 + w''$.}
\item{$\delta\in [\alpha_2, 1]$:  $r_{\cA_x} \le W$.}
\end{itemize}

Computing the expectation gives:
\begin{align}
E[R^X_{\cA_x}]&\le \int_0^{\alpha_1} \frac{W}{opt} d\delta +
 \int_{\alpha_1}^{\alpha_2} {\frac{W/2 + w''}{opt}} d\delta + \int_{\alpha_2}^1 \frac{W}{opt} d\delta \nonumber\\
&= \frac{1}{opt}\left(W +\alpha_1(W/2-w'') - \alpha_2(W/2-w'') \right)\nonumber\\
&=\frac{1}{opt}\left(W+(W/2-w'')(\log_x \frac{W}{W/2+w''} - 1)\right)\label{E:thelastcase},
\end{align}
where we used $\alpha_1=\log_x W - 1 - i$ and $\alpha_2=\log_x(W/2+w'')-i$.

Now, if $w''\le w'$, then $opt=W/2+w''$. Plugging this value in (\ref{E:thelastcase}), we obtain the expression $(Z-1)\log_x Z + 1$ with $Z=\frac{W}{W/2+w''}$, which is a decreasing function of $w''$, so it is maximized for $w''=0$, in which case we have $E[R^X_{\cA_x}]\le 1+\log_x 2$.

On the other hand, if $w''>w'$, then $opt=W/2+w'=W/2 + W/2-x^{i+\alpha_1} = W-W/x$. Let us optimize the numerator of (\ref{E:thelastcase}) as a function of $w''$. The derivative is  
\[
-\log_x \frac{W}{W/2+w''}+1 + \frac{W/2-w''}{(W/2+w'')\ln x} = 1- \frac{1}{\ln x}\left(\frac{W}{W/2+w''} + \ln \frac{W}{W/2+w''}\right)
\]
Recall that $w''> w' = W/2-W/x$, so the expression above is at least $1-\frac{1}{\ln x}\left(\frac{1}{1-1/x} + \ln \frac{1}{1-1/x}\right)$ and is positive for all $x>3$, showing that the bound we obtained is an increasing function of $w''$ and is maximized when $w''=W/2$. Plugging the values of $opt$ and $w''=W/2$ in (\ref{E:thelastcase}), we obtain that $E[R^X_{\cA_x}] \le \frac{1}{1/2+1/x}$.

Thus, covering all cases, we obtain the bound $E[R^X_{\cA_x}]\le \max\{g(x), h(x), f(x)\} + O(W_n^{-1})$, where $g(x)\overset{def}{=} 2 +\frac{2}{x\ln x}-\frac{1}{\ln x}$, $h(x)=1+\log_x 2$, and $f(x)=\frac{1}{1/2+1/x}$.
 It can be shown as in Thm.~\ref{T:identicalupper} that $x_{min}=-2W_{-1}(-\frac{1}{2e})\approx 5.3567$ minimizes $g(x)$ and $g(x_{min})\approx 1.627$. It can also be checked that $\max\{h(x_{min}),f(x_{min})\}<1.5$.
\end{proof}

The competitive ratio $1.627$ is tight for algorithm $\cA_x$. This is achieved on sequences consisting of $W_n/2$ 
unit weight requests followed by a request of weight $W_n/2$. Furthermore, using a similar idea as for the unweighted case, 
a barely-random algorithm that uses only a single bit can be obtained with competitive ratio $1.75$ (analysis omitted).

Last, it can be seen that on sequences with exponentially increasing weights, no algorithm can achieve a competitive 
ratio better than $1.5$.

\begin{theorem} \label{T:weightedlower}
Every randomized algorithm for \textsc{Part} on sequences with arbitrary weights has an expected competitive ratio of at least $1.5$. 
\end{theorem}
\begin{proof}
 Let $S_i = 2^0, 2^1, 2^2, \dots, 2^{i-1}$ denote the exponentially increasing sequence of length $i$, let $x_{\text{min}}, x_{\text{max}}$ be
 large integers with $x_{\text{max}} \ge 2 x_{\text{min}}$, and let $X = \left[ x_{\text{min}}, x_{\text{max}} \right]$. We consider the performance
 of any deterministic algorithm $\cA_{\det}$ on the uniform input distribution on set $\mathcal{S} = \{ S_i \, : \, i \in X \}$. The result for
 randomized algorithms then follows by applying Yao's lemma.
 
 Let $J$ be the set of requests at which $\cA_{\det}$ places a separator on input $S_{x_{\text{max}}}$. Note that the set of separator positions on 
 any other input of $\mathcal{S}$ is a subset of $J$. Let $I = J \cap X = \{ i_1, \dots, i_k \}$ be the separators placed within the interval $X$ (ordered
 such that $i_j < i_{j+1}$, for every $j$) and let $i_0$ be the right-most separator placed before $x_{\text{min}}$. 
 W.l.o.g., we can savely assume that $i_0 = x_{\text{min}} - 1$, since this does not worsen the algorithm (it is the optimal choice for $S_{x_{\text{min}}}$). 
 Furthermore, we can also assume that $x_k < x_{\text{max}}$, since placing a separator at request $x_{\text{max}}$ gives the worst ratio possible for $S_{x_{\text{max}}}$.
 
 We bound now the expected competitive ratio of $\cA_{\det}$, where the expectation is taken over the inputs $\mathcal{S}$. To this end, notice that 
 the optimal bottleneck value $OPT_i$ on sequence $S_i$ is $OPT_i = 2^{i-1}$. We write $R_{\cA_{\det}}^{i}$ to denote the competitive ratio of $\cA_{\det}$
 on sequence $S_i$. Then:
 
 \begin{eqnarray*}
  \Exp_{S_i \gets \mathcal{S} } R_{\cA_{\det}}^{i} & =  & \frac{1}{x_{\text{max}} - x_{\text{min}} + 1} \cdot \sum_{n = x_{\text{min}}}^{x_{\text{max}}} R_{\cA_{\det}}^{n} \ , \mbox{ and } \\   
  \sum_{n = x_{\text{min}}}^{x_{\text{max}}} R_{\cA_{\det}}^{n} & = & \underbrace{ \sum_{n = x_{\text{min}}}^{i_1 - 1} R_{\cA_{\det}}^{n}}_{I} + \sum_{n = i_1}^{i_2 - 1} R_{\cA_{\det}}^{n} + 
  \dots + \sum_{n = i_{k-1}}^{i_k - 1} R_{\cA_{\det}}^{n} + \underbrace{\sum_{n = i_k}^{x_{\text{max}}} R_{\cA_{\det}}^{n}}_{II} \ .
 \end{eqnarray*}
 We now bound $I$, $II$, and $\sum_{n = i_j}^{i_{j+1} - 1} R_{\cA_{\det}}^{n}$ for every $1 \le j \le k-1$, separately.
 \begin{enumerate}
  \item  First, observe that for every $1 \le j \le k-1$, we have $R_{\cA_{\det}}^{i_j} = \frac{2^{i_j} - 1}{2^{i_j-1}} = 2 - \frac{1}{2^{i_j - 1}}$, $R_{\cA_{\det}}^{i_j + 1} = 1$ (if $i_j + 1 < i_{j+1}$), and in general for every $a \ge 2$ with $i_j +a  < i_{j+1}$, $R_{\cA_{\det}}^{i_j+ a} = \frac{2^{i_j} + 2^{i_j+1} + \dots + 2^{i_j + a-1}}{2^{i_j + a-1}} \ge 1.5$.  Hence, $\sum_{n = i_j}^{i_{j+1} - 1} R_{\cA_{\det}}^{n}  \ge (1.5 - \frac{1}{2^{i_j-1}})(i_{j+1} - i_j) \ge (1.5 - \frac{1}{2^{x_{\text{min}}-1}})(i_{j+1} - i_j)$.
 \item  Concerning $I$, first notice that if $i_1 = x_{\text{min}}$, then $I = 0$. By similar considerations as above, if $i_1 = x_{\text{min}} + a$, then
 $I = \sum_{b=0}^{a-1} \frac{2^{x_{\text{min}}-1} + 2^{x_{\text{min}}} + 2^{x_{\text{min}}+1} + \dots + 2^{x_{\text{min}}+b-1}}{2^{x_{\text{min}} + b - 1}}$ which is at least $1.5 \cdot a$, if
 $a \ge 4$.
  \item Last, concerning $II$, this case is identical to the first case, and we can bound $II$ by $II \ge (1.5 - \frac{1}{2^{x_{\text{min}}-1}})(x_{\text{max}} - i_k)$. 
 \end{enumerate}
 Thus, we can bound $\Exp_{i \gets X} R_{\cA_{\det}}^{i}$ by:
 \begin{eqnarray*}
  \Exp_{S_i \gets \mathcal{S}} R_{\cA_{\det}}^{i} & \ge & \frac{1}{x_{\text{max}} - x_{\text{min}} + 1} \left( (x_{\text{max}} - x_{\text{min}} + 1-4) (1.5 - \frac{1}{2^{x_{\text{min}}-1}})  + 4 \cdot 1 \right) \\
  & = & 1.5 - \Order(\frac{1}{x_{\max}}) .
 \end{eqnarray*}
Thus, since $x_{\text{max}}$ can be chosen arbitrarily large, every deterministic algorithm has an expected competitive ratio of $1.5$. The result for
randomized algorithms follows by applying Yao's principle.
\end{proof}

\section{Partitioning with Arbitrary Number of Partitions} \label{sec:anyp}
We give now our results for $\textsc{Part}$ for arbitrary $p$. We first give an algorithm 
and a lower bound that directly consider \textsc{Part}. Then, we address \textsc{Part} indirectly, by first
solving \textsc{Flow}. 

\subsection{Algorithm and Lower Bound for Part: Direct Approach}

\subsubsection{Algorithm}
We give now a deterministic $2$-competitive algorithm for \textsc{Part} for any number of partitions $p$ in the preemptive online model.

A building block of our algorithm is the \textsc{Probe} algorithm, which has previously been used for tackling \textsc{Part} \cite{f91,kms97,k16}.
Algorithm \textsc{Probe} takes an integer parameter $B$, which constitutes a potential bottleneck value, and traverses the input sequence from
left to right, placing separators such that partitions of maximal size not larger than $B$ are created. 
It is easy to see that \textsc{Probe} creates at most $p$ partitions if $B \ge B^*$, the optimal bottleneck value.
\begin{lemma}\label{lem:probe}
 Let $X$ be an integer sequence. If there exists a partitioning of $X$ into $p$ parts with maximum partition weight 
 $m$, then $\textsc{Probe}(m)$ creates at most $p$ partitions.
\end{lemma}

\begin{algorithm}
\caption{$2$-Approximation for \textsc{Part} \label{alg:preemptive-two-approx}}
 \begin{algorithmic}
   \STATE $w_i \gets 0$ for all $1 \le i \le p$ \COMMENT{current partitions}
   \STATE $S \gets 0$ \COMMENT{current total weight}
   \STATE $m \gets 0$ \COMMENT{current maximum}
   \WHILE{request sequence not empty}
     \STATE $x \gets$ next request
     \STATE $S \gets S + x$, $m \gets \max \{ m, x \}$
     \STATE $B \gets 2 \cdot \max \{ m, S / p \}$ \COMMENT{bottleneck val.}
     \STATE Run \textsc{Probe}$(B)$ on $w_1, w_2, \dots, w_p, x$ and update partition weights $w_1, \dots, w_p$ \label{line:192}
   \ENDWHILE 
 \end{algorithmic}
\end{algorithm}

Instead of running \textsc{Probe} directly on the input sequence, we will run \textsc{Probe} for a carefully chosen bottleneck value $B$ on 
the sequence $w_1, \dots, w_p, x$ of current partition weights $w_1, \dots, w_p$ followed by the value of the current request $x$. 
We will prove that this run of \textsc{Probe} creates at most $p$ partitions. If \textsc{Probe} places a subsequence 
$w_i, \dots, w_j$ of partition weights into the same partition, then the partition separators between current partitions $i, \dots, j$ are removed.
See Algorithm~\ref{alg:preemptive-two-approx} for details.

\begin{theorem} \label{thm:generalpub1}
 Algorithm~\ref{alg:preemptive-two-approx} is a deterministic $2$-competitive algorithm for \textsc{Part}.
\end{theorem}
\begin{proof}
 First, suppose that the run of \textsc{Probe} Algorithm~\ref{alg:preemptive-two-approx} succeeds in 
every iteration. Let $S$ be the weight of the entire input sequence, and let $m$ be its maximum. Since the optimal bottleneck 
value $B^*$ is trivially bounded from below by $\max \{m, S/p \}$, and the bottleneck value employed in the last run of 
\textsc{Probe} 
is $B = 2 \cdot \max \{ m, S / p \}$, we obtain an approximation factor of $2$.

Denote $w_{p+1}=x$. It remains to prove that the run of \textsc{Probe} always succeeds, i.e., in every iteration of the algorithm, the 
optimal bottleneck value of the sequence $w_1, w_2, \dots w_p, w_{p+1}$ is at most $B = 2 \cdot \max \{ m, S / p \}$ (where $S$ is the 
current total weight of the input sequence, and $m$ the current maximum). Indeed, if \textsc{Probe}$(B)$ does not succeed in 
creating $p$ partitions, then $w_i + w_{i+1}> B \ge 2 S / p$ must hold for all $1 \le i \le p$. But then:
$$
 S  = \sum_{i=1}^{p+1} {w_i} \ge \sum_{i=1}^{\lfloor (p+1)/2 \rfloor} \left( w_{2i - 1} + w_{2i} \right) > \lfloor (p+1)/2 \rfloor \cdot 2 S / p \ge S,
$$
a contradiction, which proves correctness of the algorithm.
\end{proof}

\subsubsection{Lower Bound}
Next, we present a lower bound for \textsc{Part} that makes use of the discrete properties of integers (and does not apply to \textsc{Flow}). 
For a given $p \ge 2$, let $\sigma_1$ and $\sigma_2$ be the all-ones sequences of lengths $2p$ and $2p+1$, respectively. Consider first a
deterministic algorithm $\cA_{\text{det}}$. If $\cA_{\text{det}}$ computes an optimal solution on $\sigma_1$, i.e., a partitioning
consisting of $p$ partitions each of weight $2$, then $\cA_{\text{det}}$ on $\sigma_2$ is at least $\frac{4}{3}$-competitive, since 
the merging of any two partitions creates a bottleneck value of $4$ while the optimal bottleneck is $3$. On the other hand, 
if $\cA_{\text{det}}$ is not optimal on $\sigma_1$, then $\cA_{\text{det}}$ is at least $\frac{3}{2}$-competitive on $\sigma_1$. Thus, 
every deterministic algorithm for \textsc{Part} is at least $\frac{4}{3}$-competitive.

Consider now a randomized algorithm $\cA_{\text{rand}}$. Consider further the input distribution over $\{ \sigma_1, \sigma_2 \}$ so that 
$\sigma_1$ occurs with probability $\frac{2}{5}$, and $\sigma_2$ occurs with probability $\frac{3}{5}$. Let $p$ denote the 
probability that $\cA_{\text{rand}}$ outputs an optimal solution on $\sigma_1$. Then, the approximation ratio of $\cA_{\text{rand}}$ is bounded from below by:
$$
  \Pr \left[ \sigma_1 \mbox{ occurs} \right] \cdot \left( p \cdot 1 + (1-p) \cdot \frac{3}{2} \right) + \left[ \sigma_2 \mbox{ occurs} \right] \cdot \left( p \cdot \frac{4}{3} + (1-p) \cdot 1 \right) = \frac{6}{5} \ .
$$
We thus established the following theorem:
\begin{theorem} \label{thm:lbpart}
 Every deterministic (randomized) preemptive online algorithm for \textsc{Part} is at least $\frac{4}{3}$-competitive (resp. $\frac{6}{5}$-competitive).
\end{theorem}

\subsection{Algorithm and Lower Bound for Part: Indirect Approach Via Flow} \label{sec:flow}
\subsubsection{Algorithm}
\begin{algorithm}
 \begin{algorithmic}
 \REQUIRE Integer $p$, real number $x$
 \STATE $S_i \gets $ super-partition with base $i$, $\forall i \in [p]$
 \STATE $j \gets 1$
 \STATE $X_j \gets S_1$ \COMMENT{Initial conf. equals $S_1$}
 \FOR{$p-1$ times}
  \FOR{$i \gets 1 \dots p$}
   \STATE merge-next($S_i$)
   \STATE $X' \gets $ length $p$ prefix of $S_1, S_2, \dots, S_p$ 
   \IF{$X' \neq X_j$} \label{line:023}
    \STATE $j \gets j + 1$, $X_j \gets X'$
   \ENDIF
  \ENDFOR
 \ENDFOR
\end{algorithmic}
\caption{Periodic Scheme}\label{A:genflowp}
\end{algorithm}

We will first give a deterministic algorithm for \textsc{Flow}, which is roughly $1.68$-competitive, and then prove a lower bound of $1.08$.

In the following, we assume that $p$ is a power of two. We first give a scheme, which is based on a positive real number $x$ 
(we will set $x=2$ later), that defines a sequence of $p$-ary vectors $X_1, X_2, \dots$. Then, we will prove in Lemma~\ref{lem:conflow} 
how this scheme can be applied to \textsc{Flow}. To define our scheme, we require the concept of a {\em super-partition}:
\begin{definition}[Super-partition]
 For an integer $1 \le b \le p$ denoted the base, let $S = S_1, \dots, S_p$ with $S_1 = x^{\frac{b}{p}}$ 
 and $S_{i+1} = x^\frac{1}{p} S_{i}$ be the initial configuration of the super-partition. The super-partition
 evolves by merging in each step the pair of partitions whose sum is minimal, thus decreasing the length 
 of the super-partition by $1$ in every step (since we assumed that $p$ is a power of two, 
 overall we do $\log n$ merge sweeps from left to right).
\end{definition}

Our scheme is best explained via the evolution of $p$ super-partitions with different bases and is depicted in Algorithm~\ref{A:genflowp}. 
It outputs a sequence $X_1, X_2, \dots$, which corresponds to the evolution of the weights of the $p$ partitions in an algorithm for $\textsc{Flow}$.

To illustrate the scheme, as in Table~\ref{tab:flowscheme} in the introduction, we consider the case $p = 4$. Initially, four
super-partitions are in their initial state, and since $X_1$ equals the length $4$ prefix of the values of the super-partitions $S_1, S_2, \dots, S_p$,
we have that $X_1$ equals $S_1$ (in the following, the $X_i$ are highlighted in bold):
$$
 \overbrace{\bm{x^{\frac{1}{4}}, x^{\frac{2}{4}}, x^{\frac{3}{4}}, x^{\frac{4}{4}}}}^{S_1} \quad | \quad \overbrace{x^{\frac{2}{4}}, x^{\frac{3}{4}}, x^{\frac{4}{4}}, x^{\frac{5}{4}}}^{S_2} \quad | \quad \overbrace{x^{\frac{3}{4}}, x^{\frac{4}{4}}, x^{\frac{5}{4}}, x^{\frac{6}{4}}}^{S_3} \quad | \quad \overbrace{x^{\frac{4}{4}}, x^{\frac{5}{4}}, x^{\frac{6}{4}}, x^{\frac{7}{4}}}^{S_4}
$$
Then, merge-next advances super-partition $S_1$ into its next state, by merging the two lightest weights ($x^{\frac{1}{4}}$ and $x^{\frac{2}{4}}$ to $x^{\frac{1}{4}} + x^{\frac{2}{4}} = x^{\frac{1}{4}}(1 + x^{\frac{1}{4}})$). This makes the weight $x^{\frac{2}{4}}$ of super-partition $S_2$ advance to position $4$ and is thus included in $X_2$:
$$
 \bm{x^{\frac{1}{4}}(1+x^{\frac{1}{4}}), x^{\frac{3}{4}}, x^{\frac{4}{4}}} \quad | \quad \bm{x^{\frac{2}{4}}}, x^{\frac{3}{4}}, x^{\frac{4}{4}}, x^{\frac{5}{4}} \quad | \quad x^{\frac{3}{4}}, x^{\frac{4}{4}}, x^{\frac{5}{4}}, x^{\frac{6}{4}} \quad | \quad x^{\frac{4}{4}}, x^{\frac{5}{4}}, x^{\frac{6}{4}}, x^{\frac{7}{4}}
$$
Next, super-partition $S_2$ is advanced, which gives $X_3$:
$$
 \bm{x^{\frac{1}{4}}(1+x^{\frac{1}{4}}), x^{\frac{3}{4}}, x^{\frac{4}{4}}} \quad | \quad \bm{x^{\frac{2}{4}} (1 + x^{\frac{1}{4}})}, x^{\frac{4}{4}}, x^{\frac{5}{4}} \quad | \quad x^{\frac{3}{4}}, x^{\frac{4}{4}}, x^{\frac{5}{4}}, x^{\frac{6}{4}} \quad | \quad x^{\frac{4}{4}}, x^{\frac{5}{4}}, x^{\frac{6}{4}}, x^{\frac{7}{4}}
$$
Note that in the next two steps, super-partitions $S_3$ and $S_4$ are advanced, which does not affect
the first four positions. Hence, in these iterations, the condition in Line~\ref{line:023} evaluates to false and no new $X_j$ values are created. 
The entire evolution is illustrated in Appendix~\ref{app:completeschemefour}.

\mypara{Analysis}
Throughout the analysis, we use the notation $X_i[j]$ to denote the $j$th element of $X_i$. We also fix $x = 2$. 
Some of our results are stated for $x=2$ while others use a general $x$ for convenience. We also use the 
abbreviation $\alpha=x^{1/p} = 2^{1/p}$. We first show how our scheme can be used to obtain an algorithm for \textsc{Flow}. 
\begin{lemma}\label{lem:conflow}
The scheme of Algorithm~\ref{A:genflowp} gives rise to a deterministic algorithm for $\textsc{Flow}$ with competitive ratio $\frac{p+4}{p} \cdot \max_{j} \frac{\max X_j}{\avg X_j}$.
\end{lemma}
\begin{proof}
W.l.o.g., we assume that the warm-up period in $\textsc{Flow}$ is time $\sum_{x \in S_1} x$ (if not, then we scale our scheme). We initialize
the $p$ partitions with $X_1$. The transition from $X_i$ to $X_{i+1}$ in our scheme should be understood as a two-step process: 
First, the merge operation takes place (i.e., the advancing of some super-partition), which either creates an empty partition to the right, or
the last partition has been merged with more incoming flow. In case an empty partition is created, it is then filled with flow up to value 
$X_{i+1}[p]$. When our scheme terminates but the flow has not yet ended, then we simply repeat the scheme,
where all values are scaled. To see that this is possible, let $X_{\text{max}}$ be the last partition created by Algorithm~\ref{A:genflowp}, and observe that 
$X_{\text{max}}[i] = x^{\frac{i-1}{p}} \sum_{j=1}^p S_1[j]$, while $X_1[i] = x^{\frac{i}{p}}$. Hence, the final configuration is a scaled version 
of the initial configuration (by factor $x^{-\frac{1}{p}} \sum_{j=1}^p S_1[j]$). 

It is therefore enough to assume that the flow stops during the first iteration of the scheme.
At that moment, the last partition may not be entirely filled. Consider a final partitioning such that the first $p-1$ partitions 
coincide with the weights of the first $p-1$ values of some $X_i$, and the last partition is of some arbitrary weight $y \le X_i[p]$. 
Then, the competitive ratio is bounded by:
$$ \frac{ \max \{X_i[1], \dots, X_i[p-1], y \} }{\frac{1}{p} \left( y + \sum_{j=1}^{p-1}X_i[j] \right) } \le \frac{ \max X_i }{\avg X_{i-1} } \le \frac{p+4}{p} \cdot \frac{\max X_i}{\avg X_i} \ ,$$ 
where we used $\avg X_{i-1} \ge \frac{p+4}{p} \avg X_i$, which will be proved below. 
\end{proof}
The next lemma was used in the proof of the previous lemma.
\begin{lemma}\label{lem:avg}
 Let $x = 2$. Then, for every $i$: 
 $$\avg X_{i-1} \ge \frac{p}{p+4} \avg X_i .$$
\end{lemma}
\begin{proof}
First, by investigating the structure of $X_i$, it can be seen that $\max X_i \le 2 \cdot \min X_i$, for every $X_i$,
and $\max X_i \le 2 \max X_{i-1}$. These two bounds give $\max X_i \le 4 \min X_{i-1}$, which implies $\max X_i \le 4 \avg X_{i-1}$.
Then,
 \begin{eqnarray*}
  \avg X_{i} & = & \frac{1}{p} \cdot \sum_{x \in X_i} x = \frac{1}{p} \cdot \left( X_i[p] + \sum_{y \in X_{i-1}} y \right) \le \frac{1}{p} \cdot \left( 4 \avg X_{i-1} + \sum_{y \in X_{i-1}} y \right) \\
  & = & \frac{p+4}{p} \avg X_{i-1} .
 \end{eqnarray*}
\end{proof}

Thus, in order to obtain a good algorithm for $\textsc{Flow}$, we need to bound the max-over-average ratio 
$\frac{\max X_j}{\avg X_j}$, for every $X_j$, of the scheme. To this end, observe that $X_i$ is of the form:
$$X_i = \underbrace{S_1 S_2 \dots S_j}_{\mbox{Length $l$}} [ \underbrace{S_{j+1} \dots S_k}_{\mbox{Length $l+1$}} ],$$
where $S_1, \dots, S_j$ are of length $l$, $S_{j+1} \dots, S_k$ are of length $l+1$, $S_{j+1} \dots, S_k$ may or may not exist,
and the last super-partition of $X_i$ may be incomplete (not entirely included in $X_i$). 
Using Lemma~\ref{O:basicmaxavg}, it can be seen that for base $x = 2$, the maximum of 
$X_i$ always lies in either the right-most super-partition of length $l$ that is entirely included in $S_i$, or in the incomplete super-partition
of length $l$ (if it exists). 

In general, in order to describe the current $X_i$, let $L, t, m$ be such that $L$ is a power of two, 
$1\le t \le L/2$, all super-partitions have lengths either 
$l=L-t$ or $l+1=L-t+1$, and $m$ is the rightmost index of the super-partition that has length $L-t$ and is entirely included in $X_i$. 
Such a triple $(L,t,m)$ completely describes the state of $X_i$.

First, we need an expression for $\max X_i$. This quantity depends on the current lengths of the super-partitions and the value $m$, and using the quantities $L$ and $l$,
we obtain the following lemma: 

Our analysis requires a structural result on super-partitions, which we give first. In the following lemma, we relate the maximum value of a given super-partition 
to its length, and we give an expression for the total weight of a super-partition.

\begin{lemma}\label{O:basicmaxavg}
Fix a super-partition $S_i$ and suppose that its current length is $l=L-t$, where $L$ is a power of two and $1 \le t \le L/2$. Let $max_i$ denote the 
maximum at that moment and $sum_i$ the total weight of $S_i$. Then:
$$
sum_i = \frac{\alpha^i}{\alpha-1} , \quad \text{ and } \quad max_i = \alpha^{i}\cdot \alpha^{2tp/L}\cdot \frac{1-\alpha^{-2p/L}}{\alpha-1}.
$$
\end{lemma}
\begin{proof}
The first claim is easy: We have 
$$sum_1 = \sum_{j=1}^p \alpha^j=\frac{\alpha^{p+1}-\alpha}{\alpha-1}=\frac{\alpha}{\alpha-1}, \quad \mbox{ and } \quad sum_i=\alpha^{i-1}\cdot sum_1 \ ,$$
which gives the result. 

Denote by $s_{i,j}$ the $j$th element of $S_i$. Concerning the maximum, suppose first that super-partition $S_1$ is of length $L$. 
Then, the $k$th element of $S_1$ is $s_{1,k}=\alpha^{(k-1)p/L}\cdot s_{1,1}= \alpha^{(k-1)p/L}\cdot \sum_{j=1}^{p/L}\alpha^j=\alpha^{(k-1)p/L}\cdot \frac{\alpha^{p/L+1}-\alpha}{\alpha-1}$.
The corresponding elements in $S_i$, when it has length $L$, are $s_{i,k}=\alpha^{i-1}\cdot s_{1,k}= \alpha^{i-1} \cdot \alpha^{(k-1)p/L}\cdot \frac{\alpha^{p/L+1}-\alpha}{\alpha-1}$. 
Consider now the general setting, when $S_i$ is of length $l=L-t$. Then:
$$
max_i = s_{i, 2t-1} + s_{i, 2t}= (1+\alpha^{p/L})\cdot \alpha^{(2t-2)p/L}\cdot \alpha^{i-1} \cdot \frac{\alpha^{p/L+1}-\alpha}{\alpha-1}=\alpha^{i}\cdot \alpha^{2(t-1)p/L}\cdot \frac{\alpha^{2p/L}-1}{\alpha-1}.
$$
\end{proof}

Bounding $\avg X_i$ is the challenging part. Based on the triple $(L,t,m)$, our analysis requires a bound on $c$, the number 
of super-partitions entirely included in $X_i$. It is easy to see that $c$ is at least $\lceil \frac{p+m-l}{l+1} \rceil$. 
This, however, introduces difficulties, since optimizing over such a function is difficult if we need a fine enough 
optimization that does not allow us to ignore rounding effects (i.e., the estimate $c \ge \frac{p+m-l}{l+1}$ is not good enough in some cases). 

In Lemma~\ref{L:maxavgratio} we thus conduct a case distinction: In the easy cases (if $c$ is large), the estimate 
$c \ge \frac{p+m-l}{l+1}$ is good enough for our purposes and the max-over-average ratio follows by a careful calculation. 
In the other cases ($c$ is small), we analyze the max-over-average ratio using {\em reference points}: A reference point is
one where all partitions included in $X_i$ have equal length which is a power of two (for example the initial configuration
is one). We then study the behavior of the scheme between two consecutive reference points and obtain a different bound.
Lemma~\ref{L:maxavgratio} is the most technical part of this analysis.

\begin{lemma}\label{L:maxavgratio}
Let $R=\frac{\max X_i}{\avg X_i}$. The following bounds hold for each $X_i$ with configuration $(L,t,m)$, with an additive error term in $O(1/p)$:
$$ 
R \le \min \left \{\frac{\ln2\cdot 2^{2t/L}(1-2^{-2/L})}{2^{\frac{p+1-l}{p(l+1)}}-1}, 
\frac{2p\cdot (2^{2/L}-1)L}{2^{1/\ln2}\ln2\cdot(p+2-L)},
\frac{4\ln2\cdot (2^{2/L}-1) \cdot 2^{-\frac{2p}{cL}}}{2^{\frac{c+2}{p}}(2^{1/p}-1) + 2^{-1/p}-2^{-c/p}}\right \} .$$

\end{lemma}
\begin{proof}
The first and second bounds will be used when there are relatively many super-partitions in $X_i$. In that case, we can afford 
ignoring the remainder part $S_{c+1}$ (when estimating the average) and also use the bound on $c$ that ignores rounding effects.

Thus, the average is estimated as follows: 
$$p\cdot \avg X_i \ge \sum_{t=1}^c sum_t=\frac{\sum_{1}^c \alpha^t}{\alpha-1}=\frac{\alpha^{c+1}-\alpha}{(\alpha-1)^2} \ .$$
Recall that the current maximum is achieved in $S_m$, so all we need is to bound the ratio $max_m/\avg X_i$. Using the expression 
for $max_m$ given in Lemma~\ref{O:basicmaxavg}, we obtain:
$$
\frac{max_m}{\avg X_i} \le \frac{p(\alpha-1)\alpha^{m-1}\alpha^{2tp/L}(1-\alpha^{-2p/L})}{\alpha^c-1}\le \frac{\ln2\cdot \alpha^{m-1}\alpha^{2tp/L}(1-\alpha^{-2p/L})}{\alpha^{\frac{p+m}{l+1}-1}-1},
$$
where we  used the approximation $p(1-\alpha)\approx \ln 2$ which holds with an error bounded by $1/p$ (which we will henceforth omit) and 
the bound $c\ge \frac{p+m}{l+1}-1$.
Differentiation shows that the right hand side is non-increasing as a function of $m$ and is maximum at $m=1$, which gives the first bound:
$$
\frac{max_m}{\avg X_i} \le \frac{\ln2\cdot \alpha^{2tp/L}(1-\alpha^{-2p/L})}{\alpha^{\frac{p-l}{l+1}}-1} \ .
$$
This can be used to obtain bounds for smaller values of $l$. For larger values, we use the bound $e^x>1+x$, applied to $\alpha^{\frac{p-l}{l+1}}-1$ in the denominator, which gives $\alpha^{\frac{p-l}{l+1}}-1=2^{\frac{p-l}{p(l+1)}}-1>\frac{\ln2\cdot (p-l)}{p(l+1)}$, which gives, after replacing $\alpha=2^{1/p}$,
$$
\frac{max_m}{\avg X_i} \le \frac{p\cdot (1-2^{-2/L})\cdot 2^{2t/L}(L-t+1)}{p-L+t} \ .
$$
A simple and crude optimization w.r.t. $t$, where we replace $t=1$ in the denominator and then optimize the numerator w.r.t. $t$ (giving $t=\left(1-\frac{1}{2\ln2}\right)L+1$), gives the second bound:
$$
\frac{max_m}{\avg X_i} \le \frac{2p\cdot (2^{2/L}-1)L}{2^{1/\ln2}\ln2\cdot(p+1-L)} \ .
$$ 

Now, consider the case when there are few super-partitions in $X_i$, i.e., $c$ is small. We need a more delicate bound on $\avg X_i$ (by losing precision when computing $\text{max} X_i$: uncertainty principle in action), so we shift the reference point to the time step when there are \emph{exactly} $c$ super-partitions in $X_i$ and no incomplete super-partitions. Let $(L, t_c, m_c)$ be the tuple describing this point.  We concentrate on the process of going from $(L, t_c, m_c)$ to $(L, t_{c+1}, m_{c+1})$.
Note that $(L-t_c)m_c+(c-m_c)(L-t_c+1)=p$, and $0< m_c< c$, so we have 
$L-\frac{p}{c}+\frac{1}{c} \le t_c \le L-\frac{p}{c}+1$.
Consider an arbitrary point $(L, t, m)$ during this process (note that $L$ is the same). What is the number $T$ of elements of $S_{c+1}$ that are in $X_i$? Note that at $(L, t_c+1, 1)$, $T=c-m_c$. Next,  each time $t$ increments, $T$ increases by at least $c$ (not counting the possible action inside $S_{c+1}$). Thus, $T\ge c-m_c + (t-t_c)\cdot c + m\ge (t-t_c)\cdot c + m +1$. This helps us estimate the sum $s_{last}$ of the elements of $S_{c+1}$ in $X_i$ at time $(L,t,m)$. Note that the length of $S_{c+1}$ is $L-t+1$, so
$$
s_{last}\ge\sum_{i=1}^{T+t-1}s_{c+1,i}=\alpha^{c-1}\cdot \sum_{i=1}^{T+t-1}s_{1,i}=\alpha^{c-1}\cdot \frac{\alpha^{T+t}-\alpha}{\alpha-1}.
$$
Now, we can bound the average
\begin{eqnarray*}
p\cdot \avg X_i & = & \sum_{i=1}^c sum_i + s_{last} \ge \frac{\alpha^{c+1}-\alpha}{(\alpha-1)^2} + \alpha^{c-1}\cdot \frac{\alpha^{T+t}-\alpha}{\alpha-1}\\
& = & \frac{1}{(\alpha-1)^2}\cdot(\alpha^{T+t+c-1}(\alpha-1) + \alpha^{c}-\alpha)\\
& \ge & \frac{1}{(\alpha-1)^2}\cdot(\alpha^m\cdot \alpha^{c(t-t_c+1)}(\alpha-1) + \alpha^{c}-\alpha),
\end{eqnarray*}
and the max-over-average ratio
\begin{eqnarray*}
\frac{max_m}{\avg X_i} & \le & \frac{p(\alpha-1)\alpha^{m}\cdot \alpha^{2tp/L}\cdot (1-\alpha^{-2p/L})}{\alpha^m\cdot \alpha^{c(t-t_c+1)}(\alpha-1) + \alpha^{c}-\alpha} \le \frac{\ln2\cdot\alpha^{c+1}\cdot \alpha^{2tp/L}\cdot (1-\alpha^{-2p/L})}{\alpha^{c-1}\cdot \alpha^{c(t-t_c+1)}(\alpha-1) + \alpha^{c}-\alpha}\\
&= & \frac{\ln2\cdot (1-\alpha^{-2p/L}) \cdot \alpha^{2tp/L}}{\alpha^{c+2}(\alpha-1)\cdot\alpha^{ct}\alpha^{-ct_c} + \alpha^{-1}-\alpha^{-c}}
\end{eqnarray*}
where we let $m$ take its maximum possible value $m=c+1$,  and used the approximation $p(\alpha-1)\approx \ln 2$. The last expression is a decreasing function of $t$, which can be checked by differentiation (the derivative is $A(2/L-c/p) + B$ for positive values $A,B$), so $t=t_c$ gives the worst-case bound
\begin{eqnarray*}
\frac{max_m}{\avg X_i} & \le & \frac{\ln2\cdot (1-2^{-2/L}) \cdot 2^{2t_c/L}}{2^{\frac{c+2}{p}}(2^{1/p}-1) + 2^{-1/p}-2^{-c/p}}  \le \frac{4\ln2\cdot (2^{2/L}-1) \cdot 2^{-\frac{2p}{cL}}}{2^{\frac{c+2}{p}}(2^{1/p}-1) + 2^{-1/p}-2^{-c/p}},
\end{eqnarray*}
where we also used the bound $t_c\le L-p/c+1$ (observed above) in the second inequality. This proves the third inequality of the lemma.
\end{proof}

Combining the previous results, we obtain the following results:
\begin{theorem}\label{thm:scheme}
For every $X_i$ in our scheme, the max-over-average ratio is bounded as:
$$\frac{\max X_i}{\avg X_i} \le \frac{\ln2}{\sqrt{2}-1}+O(1/p)\approx 1.673 + O(1/p) \ .$$
\end{theorem}
\begin{proof}
We can assume that $p$ is large, e.g. $p\ge 256$, since we checked smaller values with a computer (see below). 
With this assumption, we can use the bounds in Lemma~\ref{L:maxavgratio}. The first bound gives the required approximation ratio for small values of $l$, $1\le l\le 8$, by checking each case separately. The second bound of the lemma gives the ratio for $8\le l \le p/16$. The remaining cases are when $l>p/16$, i.e. $c\le 16$. These cases can again be individually checked for each $c$, by using the third bound given by Lemma~\ref{L:maxavgratio}.
\end{proof}
\begin{corollary}
 If $p$ is a power of two, then there is an algorithm for $\textsc{Flow}$ with competitive ratio $\frac{\ln2}{\sqrt{2}-1}+O(1/p) \approx 1.673 + O(1/p)$.
\end{corollary}
Last, similar to Corollary~\ref{cor:lowarbitraryweights}, via rounding, an algorithm for $\textsc{Flow}$ can be used for $\textsc{Part}$ while incurring an 
error term that depends on the weights of the sequence.
\begin{corollary}
 There is a deterministic preemptive online algorithm for $\textsc{Part}$ with competitive ratio $1.68 + O(1/p) + O(max/S)$, 
 where $S$ is the total weight of the sequence and $max$ is the maximum element of the sequence.
\end{corollary}

\begin{figure}
 \begin{center} 
 \includegraphics[bb=60 60 510 360, height=6cm]{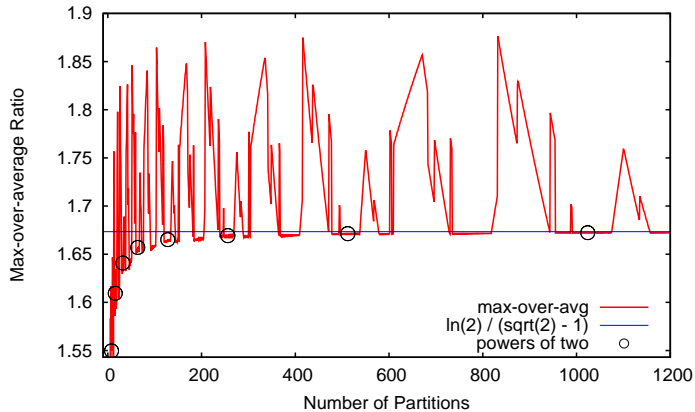}
 \caption{Extension of our scheme to arbitrary $p$. The encircled points correspond to the max-over-average ratio
 of powers of two. The blue curve is the bound proved in our analysis. \label{fig:plot}}
 \end{center} 
\end{figure}

\mypara{Extending Our Scheme to Arbitrary Number of Partitions}
Our scheme can equally be applied to arbitrary values of $p$, and we experimentally verified that max-over-average ratios
better than two are obtained. Conducting a rigorous analysis for values of $p$ that are not a power of two poses further 
complications and proves challenging. A case in point are the maxima of super-partitions, which could not be 
computed according to Lemma~\ref{O:basicmaxavg}, and many rounding problems.

In Figure~\ref{fig:plot}, we plot the max-over-average ratio of our scheme against the number of partitions $p$.
The plot shows that the power of two cases are the ones that give the best competitive ratio. The chaotic behavior
of the plot also indicates that many mechanisms within our scheme are at work at the same time that need to be bounded
appropriately. We thus pose as an open problem: Can we obtain an analyzable partitioning scheme for arbitrary
$p$ with max-over-average ratio strictly better than $2$?

\subsubsection{Lower Bound}
Consider an arbitrary current partitioning in \textsc{Flow}, and let
$x_{\text{max}}$ denote the current bottleneck value. For an appropriately chosen value $\alpha \in (1,2)$, 
partition the current partitions into two categories $L,H$, so that partitions in $L$ are of weight at most 
$\frac{x_{\text{max}}}{\alpha}$, and $H$ are all other partitions. Our argument is based on the observation
that if too many partition are in $L$, then the current competitive ratio cannot be good. However, if only
very few partitions are in $L$, then after consuming only very little flow from the input, two partitions in $H$
need to be merged with each other, which results in a large increase in the bottleneck value. This idea is
formalized in the following theorem:

\begin{theorem} \label{thm:lbflow}
Every deterministic preemptive online algorithm for \textsc{Flow} has an approximation factor of at least 
$\frac{5}{2}-\sqrt{2}\approx 1.086$. 
\end{theorem}
\begin{proof}
Let $\textsc{Alg}$ be a deterministic $f$-approximation algorithm, for some $f > 1$. Suppose that the input is a sequence of unit weight requests. 
Consider a point in time when the total weight $x_{\text{max}}$ of a heaviest partition is large enough so that rounding effects do 
not matter. Denote the weights of the partitions by $a_1,\dots,a_p$. Since $\textsc{Alg}$ is an $f$-approximation algorithm, 
\begin{equation}\label{E:lbp1}
x_{\text{max}} \le \frac{f}{p}\sum_{i=1}^p{a_i}.
\end{equation}
Let $\alpha\in (1,2)$ be a parameter whose value will be optimized later. Partition the set of indices $[p]$ 
into sets $L$ and $H$ such that $i \in L$ if $a_i < \frac{x_{\text{max}}}{\alpha}$, and $i \in H$ otherwise. 
Using this definition in Inequality~\ref{E:lbp1}, we obtain
\begin{eqnarray}
 \nonumber \frac{px_{\text{max}}}{f} & \le & |L| \cdot \frac{x_{\text{max}}}{\alpha} + (p-|L|)\cdot x_{\text{max}} \, \quad \quad \Leftrightarrow \quad \quad
|L| \le p\cdot \frac{1-1/f}{1-1/\alpha}. \label{E:lbp2}
\end{eqnarray}
We will prove now that after processing the next $(|L| + 1) \frac{2x_{\text{max}}}{\alpha}$ requests from the input sequence, 
there is at least one partition of weight at least $\frac{2x_{\text{max}}}{\alpha}$. Indeed, consider the sequence:
$$X = a_1 a_2 \dots a_p \underbrace{1 \dots 1}_{(|L|+1) \frac{2x_{\text{max}}}{\alpha}} \,$$
and, for the sake of a contradiction, suppose that there is a partitioning of this sequence into 
$p$ parts with maximum partition weight smaller than $\frac{2x_{\text{max}}}{\alpha}$. By Lemma~\ref{lem:probe}, 
such a partitioning would be found by a run of \textsc{Probe}$(\frac{2x_{\text{max}}}{\alpha} - 1)$ on $X$
read from right to left. Notice first that every resulting partition contains at most one element from $H$ (call 
such an element heavy), as the sum of two heavy elements already exceeds the bound $\frac{2x_{\text{max}}}{\alpha} - 1$. 
The first $L+1$ partitions created by \textsc{Probe} each consist of $1$s only. Since there are 
$|H| = p - |L|$ heavy elements, this implies that \textsc{Probe} creates another $|H|$ partitions 
each containing one heavy element, which totals to $p+1$ partitions, a contradiction.

Thus, after processing the next $(|L| + 1) \frac{2x_{\text{max}}}{\alpha}$ ones from the input sequence,
the average weight of a partition is bounded by:
\begin{eqnarray*}
 A & \le & \frac{|H| \cdot x_{\text{max}} + |L|\cdot \frac{x_{\text{max}}}{\alpha} + (|L|+ 1) \cdot \frac{2x_{\text{max}}}{\alpha}}{p}
 = x_{\text{max}} \cdot \frac{(p-|L|) + |L| \cdot \frac{3}{\alpha} + \frac{2}{\alpha} }{p}\\
 & = & x_{\text{max}} \left(1+|L| \frac{\frac{3}{\alpha}-1}{p} +  \frac{2}{p \alpha} \right)
\le   x_{\text{max}}\left(1+p\cdot \frac{1-\frac{1}{f}}{1-\frac{1}{\alpha}}\cdot\frac{\frac{3}{\alpha}-1}{p} +  \frac{2}{p \alpha}\right)\\
 & = & x_{\text{max}} \left( \frac{2f +\alpha -3}{f(\alpha-1)} + \frac{2}{p \alpha} \right),
\end{eqnarray*}
where we applied Inequality~\ref{E:lbp2}, and as proved above, the maximum value then is at least
$x'_{\text{max}}\ge \frac{2x_{\text{max}}}{\alpha}$. Since \textsc{Alg} is a $f$-approximation algorithm, we have $x'_{\text{max}} \le f \cdot A$, so
$$
\frac{2x_{\text{max}}}{\alpha} \le f\cdot x_{\text{max}} \left( \frac{2f +\alpha -3}{f(\alpha-1)} + \Order(1/p) \right),
$$
which implies that $2f\ge 5-\left(\alpha+\frac{2}{\alpha}\right) - \Order (1/p)$. The function $g(\alpha)=\alpha+2/\alpha$ achieves its minimum at $\alpha=\sqrt{2}$, which implies that $f\ge \frac{5}{2}-\sqrt{2} - \Order (1/p)$, giving the result.
\end{proof}
This lower bound argument shows that competitive ratios of at least $1.086$ occur repeatedly while processing the input sequence. 
Recall that the lower bound for \textsc{Part} given in Theorem~\ref{thm:lbpart} only holds for two specific input lengths.





\bibliography{kt17}

\newpage

\appendix

\section{Barely-random Algorithm for Unit Requests and $p=2$} \label{app:barelyrandom}
We show that using a single random bit gives a significant improvement over deterministic algorithms. The algorithm 
$\cA_0$ places the separator at positions $2^i$ for odd or even $i$ depending on an initial random choice. 

\begin{algorithm}[H]
\caption{$\cA_0$}
 \begin{algorithmic}
   \STATE $i\gets 0\text{ or }1$ with probability $1/2$ each
	\FOR {each request $j=1 \dots n$}
	\IF {$j = 2^{i}$}
	\STATE move the separator to the current position 
	\STATE $i\rightarrow i+2$
	\ENDIF
	\ENDFOR
 \end{algorithmic}
\end{algorithm}
Denote the competitive ratio of $\cA_0$ on sequences of length $n$ by $R_{\cA_0}^n$. Then, we obtain the following theorem:
\setcounter{theorem}{18}
\begin{theorem}
$\Exp [R_{\cA_0}^n] = 1.5.$
\end{theorem}
\begin{proof}
Let $\alpha \in [0, 1)$ and $i \in \N$ be such that the sequence length is $n=2\cdot 2^{i+\alpha}$.
The bottleneck value of an optimal partition is at least 
$\lceil\frac{n}{2}\rceil=\lceil 2^{i+\alpha}\rceil$. 
Since, $2^{i+2}>n$, the algorithm puts the separator either at position $2^{i}$ or $2^{i+1}$, each with 
probability $1/2$. In the first case, the bottleneck value is $n-2^{i}$, while in the second case it is $2^{i+1}$. 
Thus,
$$
E[R_{\cA_0}^n]=\frac{1}{2}\cdot \frac{n-2^i + 2^{i+1}}{\lceil n/2 \rceil}\le 1+\frac{1}{2^{\alpha+1}} \le 1.5 \ .
$$
\end{proof}

\section{Complete Scheme for $p=4$} \label{app:completeschemefour}
$$
 \overbrace{\bm{x^{\frac{1}{4}}, x^{\frac{2}{4}}, x^{\frac{3}{4}}, x^{\frac{4}{4}}}}^{S_1} \quad | \quad \overbrace{x^{\frac{2}{4}}, x^{\frac{3}{4}}, x^{\frac{4}{4}}, x^{\frac{5}{4}}}^{S_2} \quad | \quad \overbrace{x^{\frac{3}{4}}, x^{\frac{4}{4}}, x^{\frac{5}{4}}, x^{\frac{6}{4}}}^{S_3} \quad | \quad \overbrace{x^{\frac{4}{4}}, x^{\frac{5}{4}}, x^{\frac{6}{4}}, x^{\frac{7}{4}}}^{S_4}
$$
$$
 \bm{x^{\frac{1}{4}}(1+x^{\frac{1}{4}}), x^{\frac{3}{4}}, x^{\frac{4}{4}}} \quad | \quad \bm{x^{\frac{2}{4}}}, x^{\frac{3}{4}}, x^{\frac{4}{4}}, x^{\frac{5}{4}} \quad | \quad x^{\frac{3}{4}}, x^{\frac{4}{4}}, x^{\frac{5}{4}}, x^{\frac{6}{4}} \quad | \quad x^{\frac{4}{4}}, x^{\frac{5}{4}}, x^{\frac{6}{4}}, x^{\frac{7}{4}}
$$
$$
 \bm{x^{\frac{1}{4}}(1+x^{\frac{1}{4}}), x^{\frac{3}{4}}, x^{\frac{4}{4}}} \quad | \quad \bm{x^{\frac{2}{4}} (1 + x^{\frac{1}{4}})}, x^{\frac{4}{4}}, x^{\frac{5}{4}} \quad | \quad x^{\frac{3}{4}}, x^{\frac{4}{4}}, x^{\frac{5}{4}}, x^{\frac{6}{4}} \quad | \quad x^{\frac{4}{4}}, x^{\frac{5}{4}}, x^{\frac{6}{4}}, x^{\frac{7}{4}}
$$
$$
 \bm{x^{\frac{1}{4}}(1+x^{\frac{1}{4}}), x^{\frac{3}{4}}, x^{\frac{4}{4}}} \quad | \quad \bm{x^{\frac{2}{4}} (1 + x^{\frac{1}{4}})}, x^{\frac{4}{4}}, x^{\frac{5}{4}} \quad | \quad x^{\frac{3}{4}}(1+ x^{\frac{1}{4}}), x^{\frac{5}{4}}, x^{\frac{6}{4}} \quad | \quad x^{\frac{4}{4}}, x^{\frac{5}{4}}, x^{\frac{6}{4}}, x^{\frac{7}{4}}
$$
$$
 \bm{x^{\frac{1}{4}}(1+x^{\frac{1}{4}}), x^{\frac{3}{4}}, x^{\frac{4}{4}}} \quad | \quad \bm{x^{\frac{2}{4}} (1 + x^{\frac{1}{4}})}, x^{\frac{4}{4}}, x^{\frac{5}{4}} \quad | \quad x^{\frac{3}{4}}(1+ x^{\frac{1}{4}}), x^{\frac{5}{4}}, x^{\frac{6}{4}} \quad | \quad x^{\frac{4}{4}}(1 + x^{\frac{1}{4}}), x^{\frac{6}{4}}, x^{\frac{7}{4}}
$$
$$
 \bm{x^{\frac{1}{4}}(1+x^{\frac{1}{4}}), x^{\frac{3}{4}}(1+x^{\frac{1}{4}})} \quad | \quad \bm{x^{\frac{2}{4}} (1 + x^{\frac{1}{4}}), x^{\frac{4}{4}}}, x^{\frac{5}{4}} \quad | \quad x^{\frac{3}{4}}(1+ x^{\frac{1}{4}}), x^{\frac{5}{4}}, x^{\frac{6}{4}} \quad | \quad x^{\frac{4}{4}}(1 + x^{\frac{1}{4}}), x^{\frac{6}{4}}, x^{\frac{7}{4}}
$$
$$
 \bm{x^{\frac{1}{4}}(1+x^{\frac{1}{4}}), x^{\frac{3}{4}}(1+x^{\frac{1}{4}})} \quad | \quad \bm{x^{\frac{2}{4}} (1 + x^{\frac{1}{4}}), x^{\frac{4}{4}}(1+ x^{\frac{1}{4}})} \quad | \quad x^{\frac{3}{4}}(1+ x^{\frac{1}{4}}), x^{\frac{5}{4}}, x^{\frac{6}{4}} \quad | \quad x^{\frac{4}{4}}(1 + x^{\frac{1}{4}}), x^{\frac{6}{4}}, x^{\frac{7}{4}}
$$
$$
 \bm{x^{\frac{1}{4}}(1+x^{\frac{1}{4}}), x^{\frac{3}{4}}(1+x^{\frac{1}{4}})} \quad | \quad \bm{x^{\frac{2}{4}} (1 + x^{\frac{1}{4}}), x^{\frac{4}{4}}(1+ x^{\frac{1}{4}})} \quad | \quad x^{\frac{3}{4}}(1+ x^{\frac{1}{4}}), x^{\frac{5}{4}}, x^{\frac{6}{4}} \quad | \quad x^{\frac{4}{4}}(1 + x^{\frac{1}{4}}), x^{\frac{6}{4}}, x^{\frac{7}{4}}
$$
$$
 \bm{x^{\frac{1}{4}}(1+x^{\frac{1}{4}}), x^{\frac{3}{4}}(1+x^{\frac{1}{4}})} \quad | \quad \bm{x^{\frac{2}{4}} (1 + x^{\frac{1}{4}}), x^{\frac{4}{4}}(1+ x^{\frac{1}{4}})} \quad | \quad x^{\frac{3}{4}}(1+ x^{\frac{1}{4}}), x^{\frac{5}{4}}(1+ x^{\frac{1}{4}}) \quad | \quad x^{\frac{4}{4}}(1 + x^{\frac{1}{4}}), x^{\frac{6}{4}}, x^{\frac{7}{4}}
$$
$$
 \bm{x^{\frac{1}{4}}(1+x^{\frac{1}{4}}), x^{\frac{3}{4}}(1+x^{\frac{1}{4}})} \quad | \quad \bm{x^{\frac{2}{4}} (1 + x^{\frac{1}{4}}), x^{\frac{4}{4}}(1+ x^{\frac{1}{4}})} \quad | \quad x^{\frac{3}{4}}(1+ x^{\frac{1}{4}}), x^{\frac{5}{4}}(1+ x^{\frac{1}{4}}) \quad | \quad x^{\frac{4}{4}}(1 + x^{\frac{1}{4}}), x^{\frac{6}{4}}(1+x^{\frac{1}{4}})
$$
$$
 \bm{x^{\frac{1}{4}}(1+x^{\frac{1}{4}})(1+x^{\frac{2}{4}}) } \quad | \quad \bm{x^{\frac{2}{4}} (1 + x^{\frac{1}{4}}), x^{\frac{4}{4}}(1+ x^{\frac{1}{4}})} \quad | \quad \bm{x^{\frac{3}{4}}(1+ x^{\frac{1}{4}})}, x^{\frac{5}{4}}(1+ x^{\frac{1}{4}}) \quad | \quad x^{\frac{4}{4}}(1 + x^{\frac{1}{4}}), x^{\frac{6}{4}}(1+x^{\frac{1}{4}})
$$
$$
 \bm{x^{\frac{1}{4}}(1+x^{\frac{1}{4}})(1+x^{\frac{2}{4}}) } \quad | \quad \bm{x^{\frac{2}{4}} (1 + x^{\frac{1}{4}})(1 + x^{\frac{2}{4}})} \quad | \quad \bm{x^{\frac{3}{4}}(1+ x^{\frac{1}{4}}), x^{\frac{5}{4}}(1+ x^{\frac{1}{4}})} \quad | \quad x^{\frac{4}{4}}(1 + x^{\frac{1}{4}}), x^{\frac{6}{4}}(1+x^{\frac{1}{4}})
$$
$$
 \bm{x^{\frac{1}{4}}(1+x^{\frac{1}{4}})(1+x^{\frac{2}{4}}) } \quad | \quad \bm{x^{\frac{2}{4}} (1 + x^{\frac{1}{4}})(1 + x^{\frac{2}{4}})} \quad | \quad \bm{x^{\frac{3}{4}}(1+ x^{\frac{1}{4}})(1+x^{\frac{2}{4}})} \quad | \quad \bm{x^{\frac{4}{4}}(1 + x^{\frac{1}{4}})}, x^{\frac{6}{4}}(1+x^{\frac{1}{4}})
$$
$$
 \bm{x^{\frac{1}{4}}(1+x^{\frac{1}{4}})(1+x^{\frac{2}{4}}) } \quad | \quad \bm{x^{\frac{2}{4}} (1 + x^{\frac{1}{4}})(1 + x^{\frac{2}{4}})} \quad | \quad \bm{x^{\frac{3}{4}}(1+ x^{\frac{1}{4}})(1+x^{\frac{2}{4}})} \quad | \quad \bm{x^{\frac{4}{4}}(1 + x^{\frac{1}{4}})(1 + x^{\frac{2}{4}})}
$$

\end{document}